\documentclass{article}
\usepackage[utf8]{inputenc}
\usepackage[colorlinks=true,linkcolor=blue,citecolor=blue]{hyperref}
\usepackage{amsmath}
\usepackage{amssymb}
\usepackage{amsthm}
\usepackage{dsfont}
\usepackage{todonotes}
\usepackage{thm-restate}
\usepackage{algorithm}
\usepackage{comment}
\usepackage[noend]{algpseudocode}
\usepackage[absolute]{textpos}  % For ERC acknowledgment with flags.
\usepackage{booktabs}
\usepackage[left=3cm,right=3cm,top=1.5cm,bottom=2.5cm,includeheadfoot]{geometry}
\usepackage{authblk}

\usepackage{needspace}

\usetikzlibrary{calc}

\newtheorem{theorem}{Theorem}%[section]
\newtheorem{lemma}{Lemma}

\newtheorem{redrule}{Rule}
\newtheorem{corollary}{Corollary}

\newtheorem{cla}{Claim}
\theoremstyle{definition}
\newtheorem{definition}{Definition}

\theoremstyle{definition}
\newtheorem*{claimproof}{\normalfont{\textit{Proof}}}

\newcommand{\ko}{k_1}

\newcommand{\Oh}{\mathcal{O}}
\newcommand{\cC}{\mathcal{C}}
\newcommand{\sn}{\lceil \sqrt{n} \, \rceil}
\newcommand{\badStuffHappens}{\text{\normalfont{NP}} \subseteq \text{\normalfont{coNP}} / \text{\normalfont{poly}}}

\DeclareMathOperator{\down}{down}
\DeclareMathOperator{\up}{up}
\DeclareMathOperator{\midd}{mid}
\DeclareMathOperator{\Confl}{Confl}
\DeclareMathOperator{\Seq}{Seq}

\newif\iflong\longtrue

\newcommand{\proofparagraph}[1]{\par\smallskip\textit{#1}}

\title{Your Rugby Mates Don't Need to Know your Colleagues: Triadic Closure with Edge~Colors\footnote{A preliminary version of this work appeared in \emph{Proceedings of the 11th International Conference on Algorithms and Complexity (CIAC'19)}  held in Rome, Italy in May 2019. The long version contains full proofs of all statements, and a slightly stronger NP-hardness result for~\textsc{Multi-STC} with two strong colors.}}

\author[1]{Laurent Bulteau}
\author[2]{Niels Grüttemeier}
\author[2]{Christian Komusiewicz}
\author[3]{Manuel Sorge\footnote{Supported by the People Programme
  (Marie Curie Actions) of the European Union's Seventh Framework
  Programme (FP7/2007-2013) under REA grant agreement number
  {631163.11}, the Israel Science Foundation (grant no. 551145/14),
  and by the European Research Council (ERC)  
  under the European Union’s 
  Horizon 2020 research and innovation programme under grant
  agreement number~714704.}}

\affil[1]{CNRS, Université Paris-Est Marne-la-Vallée, Paris,
  France\\ \texttt{laurent.bulteau@u-pem.fr}}
\affil[2]{Fachbereich
Mathematik und Informatik, Philipps-Universität Marburg, Marburg,
Germany\\ \texttt{\{niegru,komusiewicz\}@informatik.uni-marburg.de}}
\affil[3]{Faculty of Mathematics, Informatics and
  Mechanics, University of Warsaw, Warsaw, Poland\\
  \texttt{manuel.sorge@mimuw.edu.pl}}
\date{}

\begin{document}

\maketitle

\begin{abstract} {Given an undirected graph~$G=(V,E)$ the NP-hard \textsc{Strong Triadic
      Closure (STC)} problem asks for a labeling of the edges as \emph{weak} and \emph{strong}
    such that at most~$k$ edges are weak and for each induced~$P_3$ in~$G$ at least one edge is
     weak. In this work, we study the following generalizations of \textsc{STC} with~$c$
    different strong edge colors. In \textsc{Multi-STC} an induced~$P_3$
    may receive two strong labels as long as they are different. In
    \textsc{Edge-List Multi-STC} and \textsc{Vertex-List Multi-STC} we may additionally
    restrict the set of permitted colors for each edge of~$G$. We show that, under the
    Exponential Time Hypothesis (ETH), \textsc{Edge-List Multi-STC} and \textsc{Vertex-List
      Multi-STC} cannot be solved in time~$2^{o(|V|^2)}$. We then \iflong proceed with a parameterized complexity analysis in which we \fi  extend previous fixed-parameter tractability results and kernelizations for STC~[Golovach
    et al., Algorithmica~'20, Grüttemeier and Komusiewicz, Algorithmica~'20] to the three variants with multiple
    edge colors or outline the limits of such an extension.}
\end{abstract}

\marginpar{\vspace{7cm}\includegraphics[width=30px]{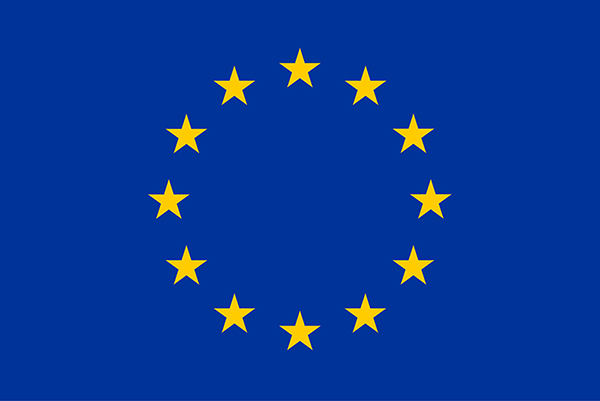}}
\marginpar{\includegraphics[width=30px]{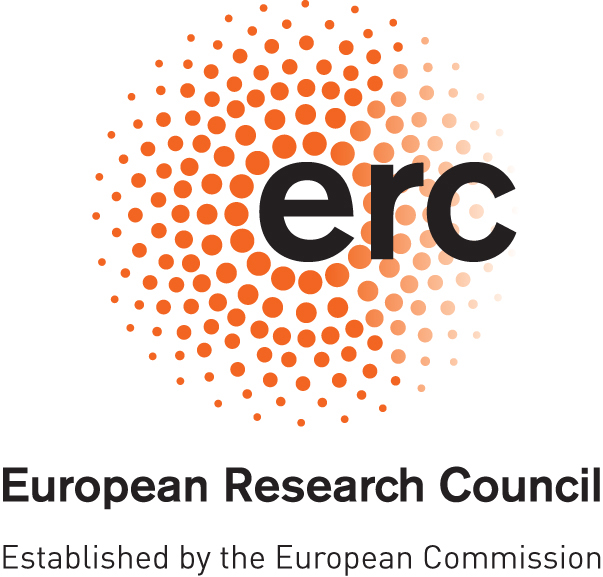}}

\subsection*{Acknowledgment}

We would like to thank the reviewers of \textit{Journal of Computer and System Sciences} for their helpful comments. In particular, one reviewer pointed out a simpler proof for Theorem~\ref{thm:allcNPh}~b).

This study was initiated at the 7th annual research retreat of the Algorithmics and
Computational Complexity group of TU Berlin, Darlingerode, Germany, March 18th--23rd,
2018.

\section{Introduction}
\label{sec:intro}
Social networks represent relationships between humans such as 
 acquaintance and friendship in online social networks. 
One task in social network analysis is to
determine the strength\iflong~\cite{Gra83,RTG17,ST14,XNR10} \else~\cite{RTG17,ST14} \fi and type~\cite{DNG07,TLK12, YL12} of the relationship signified by each 
edge of the network.
One approach to infer strong ties goes back to the notion of \emph{strong triadic closure}~\iflong due to
Granovetter~\cite{Gra73,Gra83} \else~\cite{Gra73} \fi which postulates that, if an agent has strong relations to
two other agents, then these two 
should have at least a weak relation. Following
this assertion, Sintos and Tsaparas~\cite{ST14} proposed to find strong ties in
social networks by labeling the edges as weak or strong such
that the strong
triadic closure property is fulfilled and the number of strong edges is maximized.

Sintos and Tsaparas~\cite{ST14} also formulated an extension where agents may have~$c$ different types of strong relationships. In this 
model,
the strong triadic closure property only applies to edges of the same strong type. This is
motivated by the observation that agents may very well have close relations to agents
that do not know each other if these relations themselves arise in segregated contexts. For
example, it is quite likely that one's rugby teammates do not know all of one's close
colleagues. The edge labelings\iflong\ with up to~$c$ strong colors\fi\ that model this variant of
strong triadic closure and the corresponding problem are defined as follows.

\begin{definition}   A $c$-\emph{labeling}
  $L=(S^1_L,\ldots,S^c_L,W_L)$ of an undirected graph~$G=(V,E)$ is a partition of the edge set~$E$ into~$c+1$ color classes. The edges in~$S^i_L$,~$i\in \{1, \dots, c\}$, are
   \emph{strong} and the edges in~$W_L$ are \emph{weak};
  $L$ is an \emph{STC-labeling} if there exists no pair of edges
  $\{ u, v \} \in S^i_L$ and $\{ v, w\} \in S^i_L$ such that~$\{ u, w \}\not \in E$.
  We say that such a pair of edges \emph{violates STC} for a \(c\)-labeling understood from the~context.
\end{definition}
% The problem proposed by Sintos and
% Tsaparas~\cite{ST14} is now as follows.
\begin{quote}
  \textsc{Multi Strong Triadic Closure (Multi-STC)}\\
  \textbf{Input}: An undirected graph~$G=(V, E)$ and
  integers~$c \in \mathds{N}$ and~$k \in \mathds{N}$.\\
  \textbf{Question}: Is there a $c$-colored STC-labeling~$L$ with~$|W_L| \leq k$?
\end{quote}
We  refer to the special case~$c=1$ as \textsc{Strong Triadic Closure} (STC). STC, and thus \textsc{Multi-STC}, is NP-hard~\cite{ST14}. We study the complexity of \textsc{Multi-STC} and two generalizations of \textsc{Multi-STC} which are defined as follows.

The first generalization deals with the case where one restricts the set of possible
relations for some agents. Assume, for example, that strong edges
% in the network 
correspond to family relations or professional relations. If one
% additionally
knows the profession of some agents, then this knowledge can be modeled by
introducing different strong colors for each profession and constraining the sought edge
labeling in such a way that each agent may receive only a strong edge corresponding to a
familial relation or to his profession. In other words, for
each agent we are given a list of\iflong\ allowed\fi\ strong colors that may be assigned to
incident relationships. \iflong Formally, we arrive at the following extension of STC-labelings.\fi
\begin{definition}
Let $G=(V,E)$ be a graph, $\Lambda: V \rightarrow 2^{\lbrace 1,2, \dots, c\rbrace}$ a mapping for some $c \in \mathds{N}$, and $L=(S^1_L, \dots, S^c_L,W_L)$ a $c$-colored STC-labeling. We say that an edge $\{v, w\} \in E$ \emph{satisfies the $\Lambda$-list property under $L$} if $\{ v, w \} \in W_L$ or $ \{v, w \} \in S^{\alpha}_L$ for some $\alpha \in \Lambda(v) \cap \Lambda(w)$. We call a $c$-colored STC-labeling \emph{$\Lambda$\nobreakdash-satisfying} if every edge $e \in E$ satisfies the $\Lambda$-list property under $L$.
\end{definition}
\begin{quote}
  \textsc{Vertex-List Multi Strong Triadic Closure (VL-Multi-STC)}\\
  \textbf{Input}: An undirected graph~$G=(V, E)$,
  integers~$c \in \mathds{N}$ and~$k \in \mathds{N}$, and vertex lists~$\Lambda: V \rightarrow 2^{\lbrace 1,2, \dots, c\rbrace}$.\\
  \textbf{Question}: Is there a $\Lambda$-satisfying STC-labeling~$L$ with~$|W_L| \leq k$?
\end{quote}
\textsc{Multi-STC} is the special case where~$\Lambda(v)=\{1,\dots ,c\}$ for all~$v\in V$.

The second generalization deals with the case where one restricts the set of possible relationship types for each relation. For example, if two rugby players live far apart, it is unlikely that they play\iflong\ rugby\fi\ together. We might model this knowledge by restricting the number of relationship types for this specific relationship. In other words, for each relationship we are given a list of possible strong colors that may be assigned with to this relationship.
Observe that this is an actual generalization of~\textsc{VL-Multi-STC}: Consider three agents~$v_1$, $v_2$, and~$v_3$ that are pairwise related in the network. Assume the relationship between~$v_1$ and~$v_2$ and the relationship between~$v_1$ and~$v_3$ are both restricted to `rugby' and `colleagues'. If now the relationship between~$v_2$ and~$v_3$ is restricted to `ballet class' and `drinking buddies', this situation cannot be expressed with vertex lists. This more general constraint is formalized \iflong as follows.
\begin{definition}
Let $G=(V,E)$ be a graph, $\Psi: E \rightarrow 2^{\lbrace 1,2, \dots, c\rbrace}$ a mapping for some value $c \in \mathds{N}$ and $L=(S^1_L, \dots, S^c_L,W_L)$ a $c$-colored STC-labeling. We say that an edge $e \in E$ \emph{satisfies the $\Psi$-list property under $L$} if $e \in W_L$ or $ e \in S^{\alpha}_L$ for some $\alpha \in \Psi(e)$. We call a $c$-colored STC-labeling \emph{$\Psi$-satisfying} if every edge $e \in E$ satisfies the $\Psi$-list property under $L$.
\end{definition}
This leads to the most general problem of this work.
\else 
as for vertex lists with two differences: we are given edge lists~$\Psi: E \rightarrow 2^{\lbrace 1,2, \dots, c\rbrace}$ and for each edge~$e$ we have~$e \in W_L$ or~$e \in S^{\alpha}_L$ for some $\alpha \in \Psi(e)$.
\fi
\begin{quote}
  \textsc{Edge-List Multi Strong Triadic Closure (EL-Multi-STC)}\\
  \textbf{Input}: An undirected graph~$G=(V, E)$,
  integers~$c \in \mathds{N}$ and~$k \in \mathds{N}$ and edge lists $\Psi: E \rightarrow 2^{\lbrace 1,2, \dots, c\rbrace}$.\\
  \textbf{Question}: Is there a $\Psi$-satisfying STC-labeling~$L$ with~$|W_L| \leq k$?
\end{quote}

From a more abstract point of view, in STC we need to cover all induced~$P_3$s, the paths on
three vertices, in a graph by selecting at most~$k$ edges, a natural graph-theoretic
task. Moreover, as we discuss later, all STC-problems studied here have close ties
to % \textsc{Vertex Coloring}
finding proper vertex colorings in a related graph, the Gallai graph~\cite{Gal67,Le96,Sun91} of
the input graph~$G$. Finally, in triangle-free graphs every pair of incident edges forms an induced~$P_3$. Consequently, on triangle-free graphs \textsc{Multi-STC} is equivalent to the \textsc{Edge-Colorable Subgraph} problem which asks to determine whether we can delete at least~$k$ edges from~$G$ such that the resulting graph has a proper edge-coloring with at most~$c$ colors. Hence, a study of \textsc{Multi-STC} and its two proposed generalizations is motivated not only by the known applications of
\textsc{Multi-STC} in social network analysis~\cite{ST14,HTW+14} or plausible applications of
the two generalizations, but also from a pure combinatorial
and computational complexity point of view. 

\paragraph{Related Work.}

So far, most algorithmic work has focused on STC~\cite{GHK+20,GK20,KNP18,ST14,KP17}. For
example, STC is NP-hard even on graphs with maximum degree
four~\cite{KNP18}. Motivated by this NP-hardness, the parameterized complexity of STC was
studied. The two main parameters \iflong under consideration \fi so far are the number~$k$ of weak edges and
the number~$\ell:=|E|-k$ of strong edges in an STC-labeling with a minimal number of weak
edges. The fixed-parameter tractability for~$k$ follows from a reduction to
\textsc{Vertex Cover}~\cite{ST14}. Moreover, STC admits a~$4k$-vertex kernel~\cite{GK20}. For~$\ell$,
STC is fixed-parameter tractable but does not admit a polynomial
kernel~\cite{GHK+20,GK20}. Golovach et al.~\cite{GHK+20} considered a further 
generalization of STC where the aim is to color at most~$k$ edges weak such that each induced
subgraph isomorphic to a fixed graph~$F$ has at least one weak edge. Another variant of STC
asks for a labeling in which some prespecified communities are connected via strong
edges~\cite{HTW+14,BTZ19}. In recent work, it was shown that~\textsc{Multi-STC} admits kernels for parameters that measure the distance of the input graph to a class of low-degree graphs for which \textsc{Multi-STC} can be solved in polynomial time~\cite{GKM20}.

\paragraph{Our Results.} We study the classic, fine-grained, and parameterized complexity of~\textsc{Multi-STC} and its two generalizations.
In a nutshell, we obtain
strong hardness results for \textsc{VL-Multi-STC} and \textsc{EL-Multi-STC}, showing that they cannot be
solved in~$2^{o(n^2)}$ time on $n$-vertex graphs when assuming the Exponential Time Hypothesis
(ETH)~\cite{IPZ01}. On the positive side, we show that previous fixed-parameter tractability
and kernelization results for STC~\cite{ST14,GHK+20,GK20} can be extended even to the most
general problem \textsc{EL-Multi-STC} when~$c$ is an additional parameter. In detail, we obtain the following results.

\begin{table}[t]
  \caption{An overview of the parameterized complexity results for the parameters number~$k$ of weak edges, number~$c$ of colors, and number~$k_1$ of weak edges for~$c=1$.}
  \centering
  \begin{tabular}{lccc}
    \toprule
    Parameter & ~~~~\textsc{Multi-STC}~~~~ & ~~~\textsc{VL-Multi-STC}~~~ & ~~~\textsc{EL-Multi-STC}~~~ \\
    \midrule
    $k$ & \multicolumn{3}{c}{FPT if $c\le 2$, NP-hard for~$k=0$ for all $c\ge 3$}\\
    $k_1$ & $4k_1$-vertex kernel (Corollary~\ref{Corollary: k1 kernel}) & \multicolumn{2}{c}{W[1]-hard (Theorem~\ref{VertList-W1-by-ko})}  \\
    \midrule
              & \multicolumn{3}{c}{$\Oh({(c+1)}^{k_1}\cdot (cm+nm))$ time (Theorem \ref{Theorem: (c+1)^ko Algorithm for EL-M-STC})}   \\
    $(c,k_1)$  & $4k_1$-vertex kernel (Corollary~\ref{Corollary: k1 kernel}) &  \multicolumn{2}{c}{no polynomial kernel (Corollary~\ref{Corollary: VL-Mult-STC no poly kernel})}   \\
    & &  \multicolumn{2}{c}{$2^{c+1}\ko$-vertex kernel (Corollary~\ref{Corollary: k1 kernel})}   \\\bottomrule
  \end{tabular}
  \label{tab:results}
\end{table}

First, we observe that previous results on the \textsc{Edge Coloring} problem give, for every fixed~$c$, a dichotomy of \textsc{Multi-STC} into NP-hard and polynomial-time solvable instances with respect to the maximum degree of the input graph. In particular, for all~$c\ge 3$, \textsc{Multi-STC} is NP-hard even if~$k=0$. 
\iflong For~\textsc{VL-Multi-STC} and \textsc{EL-Multi-STC}, we then show that even an algorithm that is
single-exponential in the number~$n$ of vertices of the input graph is unlikely. More
precisely, we show \else We then show \fi that, assuming the ETH, there is no~$2^{o(n^2)}$-time algorithm for
\textsc{VL-Multi-STC} and \textsc{EL-Multi-STC} even if~$k=0$ and~$c\in \Oh(n)$. \iflong This result is achieved by a compression of 3-CNF formulas~$\phi$ where each variable occurs in a constant number of clauses into graphs with~$\Oh(\sqrt{|\phi|})$ vertices. The NP-hardness of \textsc{Multi-STC} even if~$k=0$ implies that for all three problems, there is presumably no polynomial-time approximation algorithm. Furthermore, for \textsc{VL-Multi-STC} and \textsc{EL-Multi-STC} it is not even possible to compute an approximation within~$2^{o(n^2)}$~time unless ETH fails.  \fi

We then proceed to a parameterized complexity analysis \iflong for the three problems\fi; see Table~\ref{tab:results} for an overview. Since all variants are NP-hard
even if~$k=0$, we consider a structural parameter related to~$k$. This parameter,
denoted by~$k_1$, is the minimum number of weak edges needed in an STC-labeling for~$c=1$. 
Thus, if~$k_1$ is known, then we may immediately accept all instances with~$k \ge k_1$; in this sense one may assume~$k\le k_1$ for~\textsc{Multi-STC}.
For \textsc{VL-Multi-STC} and \textsc{EL-Multi-STC} this is not necessarily true since some lists might be empty which enforces weak edges.  % In our opinion, 

The
parameter~$k_1$ is relevant for two reasons: First, it allows us to determine to which extent the FPT algorithms for \textsc{STC} carry over to \textsc{Multi-STC},
\textsc{VL-Multi-STC}, and \textsc{EL-Multi-STC}.  
Second, $k_1$~has a structural interpretation: it is the vertex cover
number of the Gallai graph of the input graph~$G$. We believe that this parameterization is of
 independent interest and
might be useful for other %edge modification or labeling
problems. 
The specific results are as follows. We provide
an~$\Oh({(c+1)}^{k_1}\cdot (c|E|+|V|\cdot |E|))$-time algorithm for the most general
problem, \textsc{EL-Multi-STC}. We then use this algorithm to show that \textsc{Multi-STC}
is fixed-parameter tractable when parameterized by~$k_1$ alone. In contrast,
\textsc{VL-Multi-STC} and \textsc{EL-Multi-STC} parameterized by~$k_1$ alone are~W[1]-hard
as we show. Moreover, both problems are unlikely to admit a kernel that is polynomial
in~$c+\ko$. We do, however, obtain a kernel that is exponential in~$c$ and polynomial
in~$k$ by extending the~$4k_1$-vertex kernelization~\cite{GK20} from STC. More precisely,
we obtain a~$2^{c+1}\cdot \ko$-vertex kernel for \textsc{VL-Multi-STC} and
\textsc{EL-Multi-STC}. In the case of \textsc{Multi-STC} the kernelization gives a kernel with at most~$4k_1$ vertices, thus extending the linear-vertex kernel from~$c=1$ to arbitrary values of~$c$.  

This work is organized as follows. In Section~\ref{sec:prelim}, we present our notation and
specify the relation of STC and its variants to vertex coloring problems in Gallai graphs. This
will provide some first running-time upper bounds and explains why~$k_1$ is a natural
structural parameter. In Section~\ref{sec:hardness}, we provide the NP-hardness results and the
ETH-based lower bound for \textsc{EL-Multi-STC} and \textsc{VL-Multi-STC}. In
Section~\ref{sec:pc} we provide fixed-parameter tractability and intractability results,
including the kernelization algorithms.

\section{Preliminaries}
\label{sec:prelim}

\paragraph{Notation.}We consider undirected graphs~$G=(V,E)$ where~$n:=|V|$ denotes the number of vertices
and~$m:=|E|$ denotes the number of edges in~$G$. For a vertex~$v\in V$ we denote
by~$N_G(v):=\{u\in V\mid \{u,v\}\in E\}$ the \emph{open neighborhood of~$v$} and
  by~$N_G[v]:=N(v)\cup \{v\}$ the \emph{closed neighborhood of~$v$}. 
  For any two vertex sets $V_1,V_2 \subseteq V$, we let
$E_G(V_1,V_2) := \lbrace \lbrace v_1, v_2 \rbrace \in E \mid v_1 \in V_1, v_2 \in V_2 \rbrace$
\iflong denote the set of edges between $V_1$ and $V_2$. For any vertex set $V' \subseteq V$, we
let \fi $E_G(V') := E_G(V',V')$ be \iflong the set of edges between the vertices of $V'$. \fi We may omit
  the subscript~$G$ if the graph is clear from the context. The \emph{subgraph induced by a vertex
    set}~$S$ is denoted by~$G[S]:=(S,E_G(S))$. A subset~$S \subseteq V$ is called~\emph{vertex cover} in~$G$, if every edge~$e \in E$ has at least one endpoint in~$S$. In the~\textsc{Vertex Cover} problem, the input is an undirected graph~$G$ and an integer~$k$ and the question is, whether there exists a vertex cover of size at most~$k$ in~$G$.
%
%
  % Let $G$ be a graph, $G = (V, E)$, and
  % $c \in \mathds{N}$.
  A \emph{proper vertex coloring} with $c$ colors for some~$c\in \mathds{N}$ is a mapping
  $a \colon V \to \{1, \dots, c\}$ such that there is no edge~$\{u, v\} \in E$ with
  $a(u) = a(v)$. 
  \iflong Throughout this work we call a $c$-colored STC-labeling $L$ \emph{optimal (for a graph $G$
    and lists $\Psi$)} if $L$ is $\Psi$-satisfying and the number of weak edges $|W_L|$ is
  minimal. \fi

\paragraph{Parameterized Complexity.} In parameterized complexity~\cite{Cyg+15,DF13,FG06,Nie06} we measure the running time of algorithms depending on the total input size~$n$ and a problem parameter~$k$. A problem is called \emph{fixed-parameter tractable (FPT)} if it can be solved in~$f ( k ) \cdot \text{poly}( n )$ time for some computable function~$f$. 
An important tool in the development of parameterized algorithms is \emph{problem kernelization}. A problem kernelization is a polynomial-time preprocessing of the input data by \emph{data reduction rules}. Given an instance~$I$ with parameter~$k$, the goal is to compute an equivalent instance~$I'$ with parameter~$k' \leq k$ in polynomial time where~$I'$ has size~$g(k)$ for some computable function~$g$. The instance~$I'$ is called~\emph{problem kernel} and~$g(k)$ is called the \emph{size of the kernel}. If~$g$ is a polynomial, then the problem kernelization is called~\emph{polynomial}.

A \emph{parameterized reduction} maps any instance~$( I, k )$ of some parameterized problem~$L$ in FPT time to an equivalent instance~$( I' , k' )$ of a parameterized problem~$L'$ such that~$k' \leq f ( k )$ for some computable funktion~$f$. If the reduction can be performed in polynomial time and~$f$ is a polynomial, the parameterized reduction is a \emph{polynomial parameter transformation}. If a problem is~W[1]-hard for a parameter~$k$, then it is assumed to be not fixed-parameter tractable for~$k$. If there is a parameterized reduction from a W[1]-hard problem~$L$ parameterized by~$k$ to a problem~$L'$ parameterized by~$k'$, then the problem~$L'$ parameterized by~$k'$ is W[1]-hard. Some problems that are fixed-parameter tractable do not admit a polynomial kernel unless $\badStuffHappens$. By using polynomial parameter transformations we can transfer these kernel lower bounds to other problems~\cite{BodlaenderJK14}. The
 \emph{Exponential Time Hypothesis (ETH)} is a standard complexity theoretical conjecture used to prove lower bounds. It implies that \textsc{3\nobreakdash-CNF-SAT} cannot be solved in~$2^{o(|\phi|)}$
  time where~$\phi$ denotes the input formula~\cite{IPZ01}.

\paragraph{Gallai Graphs, \boldmath$c$-Colorable Subgraphs, and their Relation to STC.}
 \textsc{Multi-STC} can be formulated in terms of so-called \emph{Gallai graphs}~\iflong \cite{Gal67,Le96,Sun91}\else \cite{Gal67}\fi. 
 \begin{definition}
   Given a graph $G=(V,E)$, the \emph{Gallai graph} $\tilde{G}=(\tilde{V},\tilde{E})$ of~$G$ is
   defined by $\tilde{V} := E$ and
   $\tilde{E}:= \{ \{e_1, e_2\} \mid e_1 \text{ and } e_2 \text{ form an induced }P_3 \text{ in
   }G\}$.
 \end{definition}
 The Gallai graph of an~$n$-vertex and~$m$-edge graph has~$\Oh(m)$ vertices and~$\Oh(mn)$
 edges. \iflong Gallai graphs do have restricted structure in the sense that not every graph is a Gallai graph of some other graph. However, for every graph~$H$, there is a Gallai
 graph which contains~$H$ as subgraph~\cite{Le96}. \fi
 For~$c=1$ (in other words for STC), \iflong the relation to Gallai graphs is as follows: A \else a \fi graph~$G=(V,E)$ has
 an STC-labeling with at most~$k$ weak edges if and only if its Gallai graph has a vertex cover
 of size at most~$k$~\cite{ST14}. This gives an~$\Oh(1.28^k + nm)$-time algorithm by using the
 current fastest algorithm for \textsc{Vertex Cover}~\cite{CKX10}.
 More generally, a graph~$G=(V,E)$ has a $c$-colored
 STC-labeling with at most~$k$ weak edges if and only if the Gallai~graph of~$G$ has a properly $c$-colorable induced
 subgraph on~$m-k$ vertices~\cite{ST14}.

 In the following, we extend this relation to \textsc{EL-Multi-STC} by considering list-colorings
 of the Gallai graph. The special cases \textsc{VL-Multi-STC}, \textsc{Multi-STC}, and STC
 nicely embed into the construction.
 First, let us define the problem that we need to solve in the Gallai graph formally. Given a
 graph $G=(V,E)$, we call a mapping~$\chi: V \rightarrow \{0,1,\dots,c\}$ a
 \emph{subgraph-$c$-coloring} if there is no edge~$\{u,v\} \in E$ with $\chi(u)=\chi(v)\neq 0$.
 Vertices $v$ with $\chi(v)=0$ correspond to deleted vertices. The \textsc{List-Colorable Subgraph}
 problem is \iflong{}now as follows.
 \needspace{5\baselineskip}
\begin{quote}
  \textsc{List-Colorable Subgraph}\\
  \textbf{Input}: An undirected graph~$G=(V, E)$ and
  integers~$c \in \mathds{N}$,~$k \in \mathds{N}$ and lists $\Gamma: V \rightarrow 2^{\{1, \dots, c\}}$.\\
  \textbf{Question}: Is there a subgraph-$c$-coloring $\chi:V \rightarrow \{0,1,\dots, c\}$ with $|\{v \in V \mid \chi(v)=0\}| \leq k$ and $\chi(w) \in \Gamma(w)\cup \{0\}$ for every $w \in V$?
\end{quote}
\else
, given an undirected graph~$G=(V, E)$, 
  integers~$c, k \in \mathds{N}$, and lists $\Gamma: V \rightarrow 2^{\{1, \dots, c\}}$, to decide whether there is  a subgraph-$c$-coloring $\chi:V \rightarrow \{0,1,\dots, c\}$ with~$|\{v \in V \mid \chi(v)=0\}| \leq k$ and $\chi(w) \in \Gamma(w)\cup \{0\}$ for every $w \in V$.
\fi

\iflong \textsc{EL-Multi-STC} and \textsc{List-Colorable Subgraph} have the following relationship.\fi

\begin{restatable}{proposition}{Gallaieq}
  % \begin{proposition}
    \label{Proposition: Gallai equivalence}
An instance $(G,c,k, \Psi)$ of \textsc{EL-Multi-STC} is a yes-instance if and only if $(\tilde{G}, c, k, \Psi)$ is a yes-instance of \textsc{List-Colorable Subgraph}, where $\tilde{G}$ is the Gallai graph of $G$.
% \end{proposition}
\end{restatable}

\begin{proof}
To prove the proposition we first describe how to transform~$c$-colored~$\Psi$-satisfying labeling~$L$ for~$G$ into a coloring~$\chi_L$ for~$\tilde{G}$ that satisfies~$\chi_L(v) \in \Psi(v) \cup \{0\}$ for every vertex~$v$ of~$\tilde{G}$ and vice versa, such that the number of weak edges under~$L$ and the number of vertices that receive color~$0$ under~$\chi_L$ are the same. Afterwards, we show that~$L$ is an STC-labeling if and only if~$\chi_L$ is a subgraph-$c$-coloring.

 For any $c$-colored labeling~$L$ we set~$\chi_L(e):= i$ for each edge in~$S^i_L$,~$1\le i\le c$,
  and~$\chi_L(e)=0$ for each edge in~$W_L$. By definition, the $c$-colored labeling~$\chi$
  is~$\Psi$-satisfying if and only if~$\chi_L$ satisfies the list constraints in the
  \textsc{List-Colorable Subgraph} instance, that is, $\chi_L(v)\in \Psi(v)\cup \{0\}$ for each
  vertex~$v$. Moreover, the number of weak edges in~$L$ is precisely the number of vertices
  in~$\tilde{G}$ that receive color~$0$. By symmetric arguments, each
  subgraph-$c$-coloring~$\chi$ that satisfies~$\Psi$ and has $k$~vertices~$v$ such that~$\chi(v)=0$ defines a
  $c$-colored labeling~$L_\chi$ of~$G$ that is~$\Psi$-satisfying and has~$k$ weak edges.

$(\Rightarrow)$ Let~$L$ be an STC-labeling. We show that for all adjacent vertices~$u$ and~$v$
  in~$\tilde{G}$ either ~$\chi_L(u)\neq \chi_L(v)$ or~$\chi_L(u)=0$ or~$\chi_L(v)=0$. Assume
  that~$\chi_L(u)\neq 0$ and~$\chi_L(v)\neq 0$. Then, the corresponding edges~$u$ and~$v$ in~$G$ are colored 
  with some strong colors~$S^i_L$ and~$S^j_L$. Since~$u$ and~$v$ are adjacent in~$\tilde{G}$,~$u$
  and~$v$ form a~$P_3$ in~$G$ and since~$L$ is a~$c$-colored~STC-labeling, we have~$i\neq
  j$. Thus,~$\chi(u)\neq \chi(v)$.

$(\Leftarrow)$   Let~$\chi$ be a subgraph-$c$-coloring. We show that~$L_{\chi}$ is an STC-labeling. Consider a pair of incident edges~$u$ and~$v$ that form a~$P_3$
  in~$G$. If~$\chi(u)=0$ or~$\chi(v)=0$, then one of the two edges is weak
  in~$L_\chi$. Otherwise, we have~$\chi(u)\neq \chi(v)$ because~$u$ and~$v$ are adjacent
  in~$\tilde{G}$. Thus,~$L_\chi$ assigns~$u$ and~$v$ to different strong colors. Hence,~$L_\chi$
  is an STC-labeling.
\end{proof}

The correspondence from Proposition~\ref{Proposition: Gallai equivalence} means that we can solve \textsc{EL-Multi-STC} by solving \textsc{List-Colorable Subgraph} on the Gallai graph of the input graph. To this end we give a running time bound for \textsc{List-Colorable Subgraph}. The algorithm for obtaining this running time is a straightforward dynamic program over subsets. Since we
are not aware of any concrete result in the literature implying this running time bound, we
provide a proof for the sake of completeness.
\begin{restatable}{proposition}{listcolorsubg}\label{prop:elbym}
%\begin{proposition}
  \textsc{List-Colorable Subgraph} can be solved in~$\Oh(3^{n}\cdot c^2 (n+m))$~time. \textsc{EL-Multi-STC} can be solved in~$\Oh(3^m\cdot c^2 mn)$ time.
%\end{proposition}
\end{restatable}
\begin{proof}
  We define a dynamic programming table~$D$ with entries of the type~$D[S,i]$
  where~$S\subseteq V$ and~$i\in \{1, \dots, c\}$. The aim is to fill~$D$ such that for all
  entries we have~$D[S,i]= \text{`true'}$ if there is a subgraph-$c$-coloring $\chi$ for~$G[S]$
  such that~$\chi(v)\in \{1,\ldots,i\}\cap \Gamma(v)$ for all~$v\in S$
  and~$D[S,i]=\text{`false'}$ otherwise. Then, the instance is a yes-instance if and only
  if~$D[S,c] =\text{`true'} $ for some~$S$ such that~$|S|\ge n-k$. 

  The table is initialized for~$i=1$ and each~$S\subseteq V$ by setting

$$D[S,1]:=
\begin{cases}
  \text{`true'} & \text{ if } S \text{ is an independent set } \wedge \forall v\in S: 1\in \Gamma(v),    \\
  \text{`false'} & \text{ otherwise.}
\end{cases}
$$

For~$i>1$, the table entries are computed by the recurrence
$$D[S,i]:=
\begin{cases}
  \text{`true'} & \text{ if $\exists S'\subseteq S$ such that } S' \text{ is an independent set } \\ &\hspace{9.5em} \wedge~\forall v\in S': i\in \Gamma(v) \\ & \hspace{9.5em} \wedge\; D[S\setminus S',i-1] = \text{`true'},      \\
  \text{`false'} & \text{ otherwise.}
\end{cases}
$$
The correctness proof is straightforward and thus omitted. The running time is dominated by the
time needed to fill table entries for~$i>1$ and can be seen as follows. For
each~$i\in \{2,\ldots ,c\}$ we consider all partitions of~$V$ into~$S'$,~$S\setminus S'$,
and~$V\setminus S$. These are~$3^{n}$ many. For each of them, we check
in~$\Oh (c\cdot (m+n))$~time whether~$S'$ is an independent set and whether~$i\in \Gamma(v)$
for all~$v\in S'$.

The running time for \textsc{EL-Multi-STC} follows from Proposition~\ref{Proposition: Gallai equivalence} and the fact that the Gallai graph of a graph~$G$ with~$n$ vertices and~$m$ edges has~$\Oh(m)$ vertices and~$\Oh(mn)$ edges.% running time for EL-Multi-STC. 
\end{proof}

\section{Hardness Results}
\label{sec:hardness}

\subsection{NP-hardness of Multi-STC}
We first observe that \textsc{Multi-STC} is NP-hard for every fixed~$c$ even on bounded-degree graphs. It was
claimed that \textsc{Multi-STC} is NP-hard for every fixed~$c$ since in the Gallai graph this is exactly the NP-hard problem
\textsc{Odd Cycle Transversal} (in case of~$c=2$) or \textsc{Vertex $c$-Coloring} (in case of~$c \geq 3$)~\cite{ST14}. It is not known, however, whether these problems are NP-hard on Gallai graphs. 
Instead, the NP-hardness can be observed from hardness results for~\textsc{Edge Coloring}. In \textsc{Edge Coloring} one is given a graph~$G$ and a number of colors~$c$ and the question is whether we can assign the colors~$1, \dots, c$ to the edges such that no two incident edges receive the same color. This gives the following dichotomy of the complexity of \textsc{Multi-STC} on bounded-degree graphs.

\begin{theorem} \textsc{Multi-STC} exhibits the following complexity-dichotomy on bounded-degree graphs:
  % \begin{proposition}
  \label{thm:allcNPh}
  
  \begin{enumerate}
   \item[a)] For~$c=1$, \textsc{Multi-STC} is NP-hard on graphs with maximum degree at least four and solvable in polynomial time when the maximum degree is at most three.
   \item[b)]
  For $c = 2$, \textsc{Multi-STC} is NP-hard on graphs with maximum degree at least three and solvable in polynomial time when the maximum degree is at most two.
   \item[c)]
  For every $c \geq 3$, \textsc{Multi-STC} is NP-hard on instances with maximum degree at least~$c$ even if~$k=0$ and it can be solved in polynomial time if the maximum degree is at most~$c-1$.
   \end{enumerate} 	
  % \end{proposition}
 \end{theorem}

\begin{proof}
Statement~$a)$ is a known result for~\textsc{STC}~\cite{KNP18}. For statement~$c)$, the NP-hardness follows from the classic result that~\textsc{Edge Coloring} is NP-hard for every fixed~$c$ even if the input graph is triangle-free~\cite{LG83}. Furthermore, due to Vizing's Theorem, every graph with maximum degree~$\Delta$ can be edge-colored with~$\Delta+1$ colors such that no two incident edges recieve the same color~\cite{V64}. Thus, instances with maximum degree at most~$c-1$ are trivial yes-instances for~\textsc{Multi-STC}. To show statement~$b)$ we reduce from~\textsc{Edge Coloring}, which is NP-hard even if~$c=3$ and the input graph is cubic and triangle-free~\cite{H81}.

Let~$I:=(G,c)$ be an instance of~\textsc{Edge Coloring} where~$c=3$ and~$G$ is a cubic and triangle-free graph. Note that~$G$ has an even number of vertices. We define~$J:=(G,2,k)$ with~$k:=\frac{n}{2}$ and show that~$I$ is a yes-instance of \textsc{Edge Coloring} if and only if~$J$ is a yes-instance of \textsc{Multi-STC}.

$(\Rightarrow)$ Let~$I$ be a yes-instance of \textsc{Edge Coloring}. Hence, one can assign the colors~$1$, $2$, and~$3$ to the edges of~$G$ in a way that no two incident edges receive the same color. Then, the fact that~$G$ is cubic implies that each color class is a perfect matching in~$G$. We let~$L:=(S^1_L, S^2_L, W_L)$ be a labeling where~$S^1_L$ contains the edges that are assigned with color~$1$, $S^2_L$~contains the edges that are assigned with color~$2$, and~$W_L$ contains the edges that are assigned with color~$3$. Note that~$|W_L|=\frac{n}{2}=k$. Furthermore, $L$ is an STC-labeling since no two incident edges in~$G$ recieved the same color. Consequently,~$J$ is a yes-instance of \textsc{Multi-STC}.

$(\Leftarrow)$ Let~$J$ be a yes-instance of \textsc{Multi-STC}. Then, there exists an STC-labeling~$L:=(S^1_L, S^2_L, W_L)$ for~$G$ such that~$|W_L| \leq \frac{n}{2}$. Note that the strong color classes~$S^1_L$ and~$S^2_L$ each form a matching in~$G$, since~$G$ is triangle-free. We next show that the edges in~$W_L$ also form a matching. Assume towards a contradiction that there are two edges in~$W_L$ that share an endpoint. Then,~$|W_L| \leq \frac{n}{2}$ implies that there is one vertex~$v$ that is not incident with some edge in~$W_L$. Since~$G$ is cubic, this implies that~$v$ is incident with two edges of the same strong color. This contradicts the fact that~$S^1_L$ and~$S^2_L$ are matchings in~$G$. Consequently, $W_L$ is a matching. Then, since the edges of~$G$ can be partitioned into three pairwise disjoint matchings,~$I$ is a yes-instance of \textsc{Edge Coloring}.
\end{proof}

Note that, for~$c=2$, \textsc{Multi-STC} parameterized by~$k$ is fixed-parameter tractable since we can solve \textsc{Odd Cycle Transversal} in the Gallai graph~$\tilde{G}$ which is fixed-parameter
tractable with respect to~$k$~\cite{RSV04}. Thus, the dichotomy from Theorem~\ref{thm:allcNPh} is also tight with regard to the complexity for instances with constant~$k$.

\subsection{Fine-Grained Complexity}
We now provide a stronger hardness result for \textsc{VL-Multi-STC} and \textsc{EL-Multi-STC}:
we show that they are unlikely to admit a single-exponential-time algorithm with respect to the number~$n$ of vertices. Thus, the\iflong\ simple\fi\ algorithm behind Proposition~\ref{prop:elbym}  is optimal in the sense that~$m$ cannot be replaced by~$n$\iflong\ in dense graphs\fi. 

We remark that for \textsc{List-Edge Coloring} an ETH-based lower bound of~$2^{o({|V|}^2)}$ has  been shown recently~\cite{KS18}. In \textsc{List-Edge Coloring} one is given a graph~$G$, a number~$c$ of colors, and a list~$\Psi(e)$ of possible colors for each edge~$e$. The question is, whether the colors~$1,2,\dots,c$ can be assigned to the edges of~$G$ such that no two incident edges receive the same color and each edge receives a color that is on its list. While \textsc{List-Edge Coloring} and \textsc{EL-Multi-STC} with~$k=0$ correspond if the input graph is triangle-free, the construction behind the lower bound of~$2^{o({|V|}^2)}$ contains triangles with edge lists that can not be easily modeled with vertex lists. We are not aware of any direct reduction from \textsc{List-Edge Coloring} to \textsc{VL-Multi-STC} that would transfers the desired lower bound to~\textsc{VL-Multi-STC}. 

Note that there is a natural reduction that transforms an instance of \textsc{List-Edge Coloring} into an equivalent instance where the input graph is triangle-free: Replace all edges with the gadget shown in Figure~\ref{Figure: Natural Reduction Gadget}. However, this reduction adds~$\Omega(cm)$ vertices and therefore does not imply the desired lower bound since the ETH reduction for \textsc{List-Edge Coloring} outputs a graph where the number of edges is quadratic in the number of~vertices~\cite{KS18}.

\begin{figure}
\begin{center}
\begin{tikzpicture}%[scale=0.85,yscale=0.7]
\tikzstyle{knoten}=[circle,fill=white,draw=black,minimum size=5pt,inner sep=0pt]
\tikzstyle{bez}=[inner sep=0pt]

\node[knoten,label=$u$] (1)  at (0,-0.5) {};

\node[knoten, label=$v$] (2)  at (2,-0.5) {};

\node[knoten,label=$w_1$] (31) at (3,1) {};
\node[knoten,label=$w_2$] (32) at (3,0) {};
\node[bez]  at (3,-0.7) {$\vdots$};
\node[knoten, label=below:$w_{c-1}$] (33) at (3,-2) {};

\node[knoten,label=$x_1$] (41) at (5,1) {};
\node[knoten,label=$x_2$] (42) at (5,0) {};
\node[bez]  at (5,-0.7) {$\vdots$};
\node[knoten,label=below:$x_{c-1}$] (43) at (5,-2) {};

\node[knoten, label=$y$] (5)  at (6,-0.5) {};

\node[knoten, label=$z$] (6)  at (8,-0.5) {};

\draw[-, line width=1pt]  (1) to (2);

\draw[-, line width=1pt]  (2) to (31);
\draw[-, line width=1pt]  (2) to (32);
\draw[-, line width=1pt]  (2) to (33);

\draw[-, line width=1pt]  (31) to (41);
\draw[-, line width=1pt]  (31) to (42);
\draw[-, line width=1pt]  (31) to (43);
\draw[-, line width=1pt]  (32) to (41);
\draw[-, line width=1pt]  (32) to (42);
\draw[-, line width=1pt]  (32) to (43);
\draw[-, line width=1pt]  (33) to (41);
\draw[-, line width=1pt]  (33) to (42);
\draw[-, line width=1pt]  (33) to (43);

\draw[-, line width=1pt]  (5) to (41);
\draw[-, line width=1pt]  (5) to (42);
\draw[-, line width=1pt]  (5) to (43);

\draw[-, line width=1pt]  (5) to (6);
\end{tikzpicture}
\end{center}
\caption{A graph consisting of edges~$\{u,v\}$, $\{y,z\}$, $\{v,w_i\}$, $\{y,x_i\}$, and~$\{w_i,x_j\}$ for every~$i,j \in \{1, \dots, c-1\}$. It is easy to see that in every proper edge-labeling with colors~$1,\dots,c$, the edges~$\{u,v\}$ and~$\{y,z\}$ must receive the same color. The graph can thus be used as a gadget to transform an instance of~\textsc{Edge-List Coloring} into an equivalent triangle-free instance: we replace every edge~$e$ with list~$\Psi(e)$ by such gadget, set~$\Psi(\{u,v\}):=\Psi(\{y,z\}):=\Psi(e)$, and assign full lists to the other edges of the gadget.} \label{Figure: Natural Reduction Gadget}
\end{figure}

We provide a strong lower bound for~\textsc{VL-Multi-STC} that is based on a reduction from \textsc{3-SAT}. This reduction is inspired by a reduction used to show that~\textsc{Rainbow Coloring} cannot be solved in~$2^{o(n^{3/2})}$ time under the ETH~\cite{KLS18}. \textsc{Rainbow Coloring} is a mildly related problem where the input is a graph and an integer~$k$ and the question is, whether the edges can be colored with~$k$ distinct colors such that every pair of vertices is connected by a rainbow path, that is, a path where all edges on the path have distinct colors. We remark that for \textsc{Rainbow Coloring} another ETH-based lower bound of~$2^{o(m)}n^{\Oh(1)}$ has been shown recently~\cite{A17}. The compression of the variable part in our reduction works mostly analogously to the reduction to \textsc{Rainbow Coloring}. However, in \textsc{VL-Multi-STC} we have vertex lists that need to be defined carefully. For the clause part of the reduction, we use equitable colorings~\cite{HS70,KKKMS10} to achieve an even stronger compression and thus a lower bound with a quadratic function in the exponent for \textsc{VL-Multi-STC}.

\begin{restatable}{theorem}{ethbound}
%\begin{theorem}
\label{Theorem: ETH lower bound}
If the ETH is true, then \textsc{VL-Multi-STC} cannot be solved in $2^{o(|V|^2)}$ time even if restricted to instances with~$k=0$.
%\end{theorem}
\end{restatable}
\begin{proof}
  We give a reduction from \textsc{3-SAT} to \textsc{VL-Multi-STC} such
  that the resulting graph has $\mathcal{O}(\sqrt{|\phi|})$~vertices,
  where $\phi$ is the input formula and $|\phi|$ is the number of
  variables plus the number of clauses. By the Sparsification
  Lemma~\cite{IPZ01}, a~$2^{o(|\phi|)}$-time algorithm for
  \textsc{3-SAT} defeats the ETH and, hence, a~$2^{o(|V|^2)}$-time
  algorithm for \textsc{VL-Multi-STC} defeats the ETH as well.

  Below, we use $n$ for the number of variables in~$\phi$. We can
  furthermore assume that, in the formula~$\phi$, each variable occurs
  in at most four clauses, since arbitrary 3-CNF formulas can be
  transformed in polynomial time to an equivalent 3-CNF formula fulfilling
  this restriction while only increasing the formula length by a
  constant factor~\cite{Tovey84}. Observe that in such instances the
  number of clauses in $\phi$ is at most~$\frac{4}{3}n$.

Let $\phi$ be a 3-CNF formula with a set $X=\{x_1, \dots, x_n \}$ of $n$ variables and a set $\mathcal{C}:=\{C_1, \dots, C_m \}$ of $m \leq \frac{4}{3}n$ clauses. Let $C_j$ be a clause and~$x_i$ a variable occurring in $C_j$. We define the \emph{occurrence number} $\Omega(C_j,x_i)$ as the number of clauses in $\{C_1, C_2, \dots, C_j \}$ that contain~$x_i$. Note that~$\Omega(C_j,x_i)$ is only defined if~$x_i$ occurs in~$C_j$. Intuitively, $\Omega(C_j,x_i)=r$ means that the $r$th occurrence of variable~$x_i$ is the occurrence in clause~$C_j$. Since each variable occurs in at most four clauses, we have~$\Omega(C_j,x_i) \in \{1,2,3,4\}$.

We describe in three steps how to construct an equivalent instance~$(G=(V,E),c=9n+4,k=0,\Lambda)$ for \textsc{VL-Multi-STC} such that~$|V| \in \mathcal{O}(\sqrt{n})$. First, we describe the variable gadget. Second, we describe the clause gadget. In a third step, we describe how these two gadgets are connected. Before we present the formal construction, we give some~intuition.

The strong colors $1, \dots, 8n$ represent the truth assignments of the occurrences of the variables. Throughout this proof we refer to these strong colors as $T^r_i, F^r_i$ with $i \in \{1, \dots, n\}$ and $r \in \{1,2,3,4\}$. The idea is that a strong color $T^r_i$ represents a `true'-assignment and $F^r_i$ represents a `false'-assignment of the $r$th occurrence of a variable $x_i \in X$. The strong colors $8n+1, \dots, 9n+4$ are auxiliary strong colors which we need for the correctness of our construction. Throughout this proof we refer to these strong colors as~$R_1, \dots, R_n$ and~$Z_1,Z_2,Z_3,Z_4$. In the variable gadget, there are four distinct edges $e_1, e_2, e_3, e_4$ for each variable~$x_i$ representing the (at most) four occurrences of the variable $x_i$. We define vertex lists that ensure that every such edge~$e_r$ can only be labeled with the strong colors $T_i^r$ and~$F_i^r$. The coloring of these edges represents a truth assignment to the variable~$x_i$. In the clause gadget, there are $m$ distinct edges such that the coloring of these edges represents a choice of literals that satisfies $\phi$. The edges between the two gadgets make the values of the literals from the clause part consistent with the assignment of the variable part. The construction consists of five layers of vertices. In the variable gadget we have an upper layer, a middle layer, and a down layer ($U^X$,$M^X$ and $D^X$). In the clause gadget we have an upper and a down layer ($U^\cC$ and $D^\cC$). Figure~\ref{Figure ETH reduction construction appendix} shows a sketch of the construction.

\begin{figure}[t]
\begin{center}
\begin{tikzpicture}[scale=0.5]
\tikzstyle{knoten}=[circle,fill=white,draw=black,minimum size=5pt,inner sep=0pt]
\tikzstyle{sound}=[rectangle,fill=white,draw=black,minimum height=15pt, minimum width=42pt,inner sep=0pt]
\tikzstyle{bez}=[inner sep=0pt]
%UX
		\node[bez] () at (-10,0) {$U^X$\makebox[0pt][l]{ $\otimes$}};
		\draw (0,0) node[minimum height=1.1cm,minimum width=11cm,draw] {};
		\node[sound] (s1) at (-4,0) {};
		\node[sound] (s2) at (1.5,0) {};
		\node[sound] (s3) at (7.5,0) {};
		
%MX
		\node[bez] () at (-10,-4) {$M^X$\makebox[0pt][l]{ $\odot$}};
		\draw (0,-4) node[minimum height=1.1cm,minimum width=11cm,draw] {};
		\node[knoten] (g11) at (-2.6,-4) {};
		\node[bez] () at (-4.5,-4.3) {$\gamma^1_{\text{mid}^{X}(x_1)}$};
		\node[knoten] (g12) at (-2.2,-4) {};
		\node[knoten] (g13) at (-1.8,-4) {};
		\node[knoten] (g14) at (-1.4,-4) {};
		\node[bez] () at (0.5,-4.3) {$\gamma^4_{\text{mid}^{X}(x_2)}$};

		\node[knoten] (g31) at (5.4,-4) {};
		\node[bez] () at (7.3,-4.3) {$\gamma^4_{\text{mid}^{X}(x_3)}$};
		\node[knoten] (g32) at (5,-4) {};
		\node[knoten] (g33) at (4.6,-4) {};
		\node[knoten] (g34) at (4.2,-4) {};

%DX
		\node[bez] () at (-10,-8) {$D^X$\makebox[0pt][l]{ $\otimes$}};
		\draw (0,-8) node[minimum height=1.1cm,minimum width=11cm,draw] {};
		\node[knoten] (d1) at (-3.5,-8) {};
		\node[bez] () at (-5.3,-8.5) {$\delta_{\down^{X}(x_1)}$};
		\node[knoten] (d2) at (1.5,-8) {};
		\node[bez] () at (-0.3,-8.5) {$\delta_{\down^{X}(x_2)}$};
		\node[knoten] (d3) at (7,-8) {};
		\node[bez] () at (5.3,-8.5) {$\delta_{\down^{X}(x_3)}$};

%UC
		\node[bez] () at (-10,-12) {$U^{\mathcal{C}}$\makebox[0pt][l]{ $\otimes$}};
		\draw (0,-12) node[minimum height=1.1cm,minimum width=11cm,draw] {};
		\node[knoten] (e1) at (1.5,-12) {};
		\node[bez] () at (0,-12.5) {$\eta_{\text{up}^{\mathcal{C}}(C_j)}$};

%DC
		\node[bez] () at (-10,-16) {$D^{\mathcal{C}}$\makebox[0pt][l]{ $\otimes$}};
		\draw (0,-16) node[minimum height=1.1cm,minimum width=11cm,draw] {};
		\node[knoten] (t1) at (1.5,-16) {};
		\node[bez] () at (0,-16.3) {$\theta_{\down^{\mathcal{C}}(C_j)}$};

%Kanten
		\draw[very thick] (t1) to (e1);
		
		\draw[very thick, bend left] (e1) to (d1);
		\draw[very thick] (e1) to (d2);
		\draw[very thick, bend right] (e1) to (d3);
		
		\draw[very thick] (d1) to (g11);
		\draw[-] (d1) to (g12);
		\draw[-] (d1) to (g13);
		\draw[-] (d1) to (g14);
		
		\draw[-] (d2) to (g11);
		\draw[-] (d2) to (g12);
		\draw[-] (d2) to (g13);
		\draw[very thick] (d2) to (g14);
		
		\draw[very thick] (d3) to (g31);
		\draw[-] (d3) to (g32);
		\draw[-] (d3) to (g33);
		\draw[-] (d3) to (g34);
		
		\foreach \x in {-1,-0.6,...,1}{
          \draw[-, color=black!20] (g11) to ($(s1)+(1.3*\x,-0.2)$);
          \draw[-, color=black!20] (g12) to ($(s1)+(1.3*\x,-0.2)$);
          \draw[-, color=black!20] (g13) to ($(s1)+(1.3*\x,-0.2)$);
          \draw[-, color=black!20] (g14) to ($(s1)+(1.3*\x,-0.2)$);
          
          \draw[-, color=black!20] (g11) to ($(s2)+(1.3*\x,-0.2)$);
          \draw[-, color=black!20] (g12) to ($(s2)+(1.3*\x,-0.2)$);
          \draw[-, color=black!20] (g13) to ($(s2)+(1.3*\x,-0.2)$);
          \draw[-, color=black!20] (g14) to ($(s2)+(1.3*\x,-0.2)$);
          
          \draw[-, color=black!20] (g31) to ($(s3)+(1.3*\x,-0.2)$);
          \draw[-, color=black!20] (g32) to ($(s3)+(1.3*\x,-0.2)$);
          \draw[-, color=black!20] (g33) to ($(s3)+(1.3*\x,-0.2)$);
          \draw[-, color=black!20] (g34) to ($(s3)+(1.3*\x,-0.2)$);
      }
      	%Supportknoten
      	\node[bez] () at (-3.8,-6) {$T^1_1$};
      	\node[bez] () at (1,-6) {$T^4_2$};
      	\node[bez] () at (7,-6) {$F^4_3$};
      	\node[bez] () at (-3.3,-10) {$F^1_1$};
      	\node[bez] () at (2.1,-10) {$F^4_2$};
      	\node[bez] () at (6.6,-10) {$T^4_3$};
      	\node[bez] () at (4.5,-14) {$T^1_1 \in \{ T^1_1, F^4_2, T^4_3\}$};
     	\node[sound, color=black!10] (s1) at (-4,0) {};
		\node[sound, color=black!10] (s2) at (1.5,0) {};
		\node[sound, color=black!10] (s3) at (7.5,0) {};
\end{tikzpicture}
\end{center}
\caption{An example of the construction. The rectangles in~$U^X$ represent vertices~$\alpha^{(r,r')}_t$ with the same value of~$t$, $\otimes$ a clique, and $\odot$ an independent set. The edge~$\{\eta_{\text{up}^{\mathcal{C}}(C_j)}, \theta_{\down^{\mathcal{C}}(C_j)}\}$ represents a clause~$C_j=(x_1 \lor \overline{x_2} \lor x_3)$ with~$\Omega(C_j,x_1)=1$ and~$\Omega(C_j,x_2)=\Omega(C_j,x_3)=4$. The edge~$\{\gamma^1_{\text{mid}^{X}(x_1)}, \delta_{\down^{X}(x_1)} \}$ has strong color~$T^1_1$ which models an assignment where~$x_1$ is true, which satisfies~$C_j$. Note that, due to the compression, we may have~$\text{mid}(x_1)=\text{mid}(x_2)$ and therefore~$x_1$ and~$x_2$ may share the four middle vertices.}
\label{Figure ETH reduction construction appendix}
\end{figure}

\proofparagraph{The Variable Gadget.}
The vertex set of the variable gadget consist of an upper layer, a middle layer and a down layer. The vertices in the middle layer and the down layer form a \emph{variable-representation gadget}, where each edge between the two parts represents one occurrence of a variable. The vertices in the upper layer form a \emph{variable-soundness gadget}, which we need to ensure that for each variable either all occurrences are assigned `true' or all occurrences are assigned `false'. For an illustration of the variable-representation and the variable-soundness gadget for some variable~$x_i$ see~Figure~\ref{Figure Variable gadget}.

\begin{figure}[t]
\begin{center}
\begin{tikzpicture}[scale=0.8]
\tikzstyle{knoten}=[circle,fill=white,draw=black,minimum size=5pt,inner sep=0pt]
\tikzstyle{bez}=[inner sep=0pt]
%Knote
		\node[knoten] (d) at (0,-1) {};
		\node[bez] () at (0,-1.4) {$\delta_t$};
		\node[knoten] (g1) at (-6,1) {};
		\node[bez] () at (-6,0.6) {$\gamma^1_{t'}$};
		\node[knoten] (g2) at (-2,1) {};
		\node[bez] () at (-2,0.6) {$\gamma^2_{t'}$};
		\node[knoten] (g3) at (2,1) {};
		\node[bez] () at (2,0.6) {$\gamma^3_{t'}$};
		\node[knoten] (g4) at (6,1) {};
		\node[bez] () at (6,0.6) {$\gamma^4_{t'}$};
		
		\node[knoten] (a12) at (-4,2.5) {};
		\node[bez] () at (-4,2.0) {$\alpha^{(1,2)}_t$};
		\node[knoten] (a21) at (-4,3.5) {};
		\node[bez] () at (-4,3.0) {$\alpha^{(2,1)}_t$};
		
		\node[knoten] (a23) at (0,2.5) {};
		\node[bez] () at (0,2.0) {$\alpha^{(2,3)}_t$};
		\node[knoten] (a32) at (0,3.5) {};
		\node[bez] () at (0,3.0) {$\alpha^{(3,2)}_t$};
		
		\node[knoten] (a34) at (4,2.5) {};
		\node[bez] () at (4,2.0) {$\alpha^{(3,4)}_t$};
		\node[knoten] (a43) at (4,3.5) {};
		\node[bez] () at (4,3.0) {$\alpha^{(4,3)}_t$};
		
		\node[knoten] (a13) at (-2,5) {};
		\node[bez] () at (-2,4.5) {$\alpha^{(1,3)}_t$};
		\node[knoten] (a31) at (-2,6) {};
		\node[bez] () at (-2,5.5) {$\alpha^{(3,1)}_t$};
		
		\node[knoten] (a24) at (2,5) {};
		\node[bez] () at (2,4.5) {$\alpha^{(2,4)}_t$};
		\node[knoten] (a42) at (2,6) {};
		\node[bez] () at (2,5.5) {$\alpha^{(4,2)}_t$};
		
		\node[knoten] (a14) at (0,7.5) {};
		\node[bez] () at (0,7) {$\alpha^{(1,4)}_t$};
		\node[knoten] (a41) at (0,8.5) {};
		\node[bez] () at (0,8) {$\alpha^{(4,1)}_t$};
		
%KantenEbene1
		\draw[thick] (d) to (g1);
		\draw[-] (d) to (g2);
		\draw[-] (d) to (g3);
		\draw[thick] (d) to (g4);
		
		\draw[-,bend left=20] (g1) to (a12);
		\draw[-,bend left=20] (g1) to (a21);
		\draw[-,bend right=20]  (g2) to (a12);
		\draw[-,bend right=20] (g2) to (a21);
		
		\draw[-,bend left=20] (g2) to (a23);
		\draw[-,bend left=20] (g2) to (a32);
		\draw[-,bend right=20] (g3) to (a23);
		\draw[-,bend right=20] (g3) to (a32);
		
		\draw[-,bend left=20] (g3) to (a34);
		\draw[-,bend left=20] (g3) to (a43);
		\draw[-,bend right=20] (g4) to (a34);
		\draw[-,bend right=20] (g4) to (a43);

%KantenEbene2
		\draw[-,bend left=30] (g1) to (a13);
		\draw[-,bend left=30] (g1) to (a31);
		\draw[-,bend right=30] (g3) to (a13);
		\draw[-,bend right=30] (g3) to (a31);
		
		\draw[-,bend left=30] (g2) to (a24);
		\draw[-,bend left=30] (g2) to (a42);
		\draw[-,bend right=30] (g4) to (a24);
		\draw[-,bend right=30] (g4) to (a42);
		
%KantenEbene3
		\draw[-,bend left=40] (g1) to (a14);
		\draw[thick,bend left=40] (g1) to (a41);
		\draw[-,bend right=40] (g4) to (a14);
		\draw[thick,bend right=40] (g4) to (a41);

%Kantenfarben
		\node[bez] () at (-3.2,-0.5) {$\{T^1_i, F^1_i\}$};
		\node[bez] () at (3.2,-0.5) {$\{T^4_i, F^4_i\}$};
		
		\node[bez] () at (-5.5,6) {$\{F^1_i, R_i\}$};
		\node[bez] () at (5.5,6) {$\{T^4_i, R_i\}$};
\end{tikzpicture}
\end{center}
\caption{The variable-representation and the variable-soundness gadget for one variable~$x_i \in X$ such that~$\down^X(x_i)=\text{up}^X(x_i)=t$ and~$\midd^X(x_i)=t'$ with the possible colors for the edges~$\{\delta_t, \gamma^1_{t'} \}$, $\{\delta_t, \gamma^4_{t'} \}$, $\{\alpha^{(4,1)}_t, \gamma^1_{t'} \}$, and~$\{\alpha^{(4,1)}_t, \gamma^4_{t'} \}$. Note that labeling~$\{\delta_t, \gamma^1_{t'} \}$ with the strong color~$F^1_i$ and labeling~$\{\delta_t, \gamma^4_{t'} \}$ with the strong color~$T^4_i$ causes a~$P_3$ with some strong color.}
\label{Figure Variable gadget}
\end{figure}

We start by describing the variable-representation gadget.~Let 
\begin{align*}
M^X &:= \{ \gamma^r_t \mid t \in \{1, \dots, \sn \}, r \in \{1,2,3,4\} \} &\text{be the set of middle vertices, and}\\
D^X &:= \{ \delta_t \mid t \in \{1,\dots, \sn +9 \} \} &\text{be the set of down vertices.}
\end{align*}
We add edges such that $D^X$ becomes a clique in $G$. To specify the correspondence between the variables in $X$ and the edges in the variable-representation gadget, we define below two mappings $\midd^X: X \rightarrow \{1, \dots, \sn \}$ and $\down^X: X \rightarrow \{1, \dots, \sn+9 \}$. Then, for each $x_i \in X$ we add four edges $\{\gamma^r_{\midd^X(x_i)}, \delta_{\down^X(x_i)} \}$ for $r \in \{1,2,3,4\}$. We carefully define the two mappings $\midd^X$, $\down^X$ and the vertex lists $\Lambda(v)$ for every $v \in M^X \cup D^X$ of the variable-representation gadget. 

Intuitively, the chosen truth assignment for each variable will be transmitted to a clause by edges between the variable and clause gadgets. To ensure that each such transmitter edge is used for exactly one occurrence of one variable, we first define the \emph{variable-conflict graph}~$H_{\phi}^X:=(X, \Confl^X)$  by $\Confl^X:= \{ \{x_i,x_j\} \mid x_i \text{ and }x_j \text{ occur in the same clause }C \in \mathcal{C}\}$, which we use to define $\midd^X$ and $\down^X$. Since every variable of $\phi$ occurs in at most four clauses, the maximum degree of $H^X_{\phi}$ is at most $8$. Hence, there is a proper vertex 9-coloring $\chi : X \rightarrow \{1,2,\dots,9\}$ for $H^X_{\phi}$ which we compute in polynomial time by a folklore greedy algorithm. We end up with $9$ color classes $\chi^{-1}(1), \dots, \chi^{-1}(9)$. Then, we partition each color class $\chi^{-1}(i)$ into $\frac{|\chi^{-1}(i)|}{\lceil \sqrt{n}\, \rceil}$ groups arbitrarily such that each group has size at most~$\lceil \sqrt{n}\, \rceil$. Let $s$ be the overall number of such groups and let $\mathcal{S}:=\{S_1, S_2, \dots, S_s\}$ be the family of all such groups of vertices in $H^X_{\phi}$ (each corresponding to a pair of a color $i \in \{1, \ldots, 9\}$ and a group in $\chi^{-1}(i)$). The following claim is directly implied by the definition of~$\mathcal{S}$ (for part~(b) observe that at most $\lceil \sqrt{n}\, \rceil$ new groups are introduced during the partitioning of the color classes).
\begin{cla} \label{Claim: Groupsizes}
For the family $\mathcal{S}:=\{S_1, S_2, \dots, S_s\}$ of groups of vertices in $H^X_{\phi}$, it holds that
\begin{enumerate}
\item[(a)] $|S_i| \leq \lceil \sqrt{n}\, \rceil$ for each $i \in \{1,\dots,s\}$, and
\item[(b)] $s \leq \lceil \sqrt{n}\, \rceil + 9$.
\end{enumerate}
\end{cla}

For any given $x_i \in X$ we define $\down^X(x_i):=j$ as the index of the group $S_j$ that contains~$x_i$. The mapping is well-defined since $\mathcal{S}$ forms a partition of the set of variables.

\begin{cla} \label{Claim: Different variable downs}
If $x_i, x_j \in X$ occur in the same clause $C \in \mathcal{C}$, then $\down^X(x_i) \neq \down^X(x_j)$.
\end{cla}
\begin{claimproof}
By definition, $x_i$ and $x_j$ are adjacent in~$H^X_{\phi}$. Hence,~$x_i$ and~$x_j$ are in different color classes and therefore elements of different groups of~$\mathcal{S}$.$\hfill \Diamond$
\end{claimproof}
Next, we define the mapping $\midd^X: X \rightarrow \{1,\dots, \sn \}$. To this end, consider the finite sequence $\Seq^n:= (\down^X(x_1), \down^X(x_2), \dots, \down^X(x_n)) \in \{ 1, \dots, \lceil \sqrt{n}\, \rceil +9\}^n$. We define $\midd^X(x_i)$ as the number of occurrences of $\down^X(x_1)$ in the partial sequence $\Seq^i:= (\down^X(x_1), \dots, \down^X(x_i))$. From Claim \ref{Claim: Groupsizes} (a) we conclude $\midd^X(x_i) \in \{1,2,\dots, \sn \}$ for every~$x_i \in X$.

\begin{cla} \label{Claim: representing edges not equal X}
Let $x_i, x_j \in X$ and $r \in \{1,2,3,4\}$. If $x_i \neq x_j$, then $$\{\gamma_{\midd^X(x_i)}^r, \delta_{\down^X(x_i)} \} \neq \{\gamma_{\midd^X(x_j)}^r, \delta_{\down^X(x_j)} \} \text{.}$$
\end{cla}

\begin{claimproof}
Without loss of generality, assume~$i<j$. Obviously, the claim holds if $\down^X(x_i) \neq \down^X(x_j)$. Let $\down^X(x_i) = \down^X(x_j)$. Then, there is at least one more occurrence of $\down^X(x_i)$ in the partial sequence $\Seq_1^j$ compared to $\Seq_1^i$. Therefore,~$\midd^X(x_i) \neq \midd^X(x_j)$.
$\hfill \Diamond$
\end{claimproof}

Thus we assigned a unique edge in~$E(M^X, D^X)$ to each occurrence of a variable in~$X$. Furthermore, the assigned edges of variables that occur in the same clause do not share an endpoint in~$D^X$ (Claim~\ref{Claim: Different variable downs}).

We complete the description of the variable-representation gadget by defining the vertex list~$\Lambda(v)$ for every~$v \in M^X \cup D^X$. We set
\begin{align*}
\Lambda(\gamma^r_t) &:=  \bigcup_{\substack{x_i \in X \\ \midd^X(x_i)=t}} \{T_i^r, F_i^r, R_i\} &\text{ for every } \gamma^r_t \in M^X \text{, and}\\
\Lambda(\delta_t) &:= \bigcup_{\substack{x_i \in X \\ \down^X(x_i)=t}} \{T_i^1,T_i^2,T_i^3,T_i^4,F_i^1,F_i^2,F_i^3,F_i^4 , Z_2\} &\text{ for every } \delta_t \in D^X\text{.}
\end{align*}

\begin{cla} \label{Claim: Edge Colors Variable Representer}
Let $x_i \in X$ and $r \in \{1,2,3,4\}$. Then, $\Lambda(\gamma^r_{\midd^X(x_i)}) \cap \Lambda(\delta_{\down^X(x_i)}) = \{T_i^r, F_i^r \}$.
\end{cla}

\begin{claimproof}
Let $\Lambda(i,r):=  \Lambda(\gamma^r_{\midd^X(x_i)}) \cap \Lambda(\delta_{\down^X(x_i)})$. Obviously, $T_i^r$ and~$F_i^r$ are elements of~$\Lambda(i,r)$. It remains to show that there is no other strong color $Y \in \Lambda(i,r)$.

\textbf{Case 1:} $Y=Z_2$. Then, $Z_2 \not \in \Lambda(\gamma^r_t)$ and it follows that~$Y \not \in \Lambda(i,r)$.

\textbf{Case 2:} $Y=R_j$ with $j \in \{1, \dots, n\}$. Then, $R_j \not \in \Lambda(\delta_t)$ and it follows that~$Y \not \in \Lambda(i,r)$.

\textbf{Case 3:} $Y=T^{r'}_{j}$ or~$Y=F^{r'}_{j}$ with~$r' \neq r$ and~$j \in \{1, \dots, n\}$. Then,~$Y \not \in \Lambda(\gamma^r_{\midd^X(x_i)})$ and it follows that~$Y \not \in \Lambda(i,r)$.

\textbf{Case 4:} $Y=T^{r}_{i'}$ or $Y=F^{r}_{i'}$ with $i' \neq i$. Assuming $T^r_{i'} \in \Lambda(i,r)$ it follows from the definition of $\Lambda$ that there is some variable $x_{i'} \neq x_i$ such that $\down^X(x_{i'}) = \down^X(x_i)$ and $\midd^X(x_{i'}) = \midd^X(x_i)$, which contradicts Claim \ref{Claim: representing edges not equal X}. Hence, $Y \not \in \Lambda(i,r)$.$\hfill \Diamond$
\end{claimproof}

Note that for each variable $x_i$ there are four edges $\{\{\gamma^r_{\midd^X(x_i)}, \delta_{\down^X(x_i)} \} \mid r \in \{1,2,3,4\}\}$ that can only be colored with the strong colors $T^r_i$ and $F^r_i$ representing the truth assignments of the four occurrences of variable $x_i$. We need to ensure that there is no variable $x_i$, where, for example, the first occurrence is set to `true' ($T^1_i$) and the second occurrence is set to `false' ($F^2_i$) in a~$\Lambda$-satisfying STC-labeling with no weak edges. To this end, we use a variable-soundness gadget, which we describe in the following.

Define 
$$U^X := \{ \alpha^{(r,r')}_t  \mid t \in \{1, \dots, \sn+9\}, (r,r') \in \{1,2,3,4\}^2, r \neq r' \}$$
to be the set of upper vertices. We add edges such that the vertices in $U^X$ form a clique in~$G$. To specify the correspondence between the variables and the edges in the variable-soundness gadget, we define below a mapping $\up^X: X \rightarrow \{1,2,\dots,\sn+9\}$. The main idea of the variable-soundness gadget is that for each variable $x_i \in X$ and each pair $\{r,r'\} \subseteq \{1,2,3,4\}$ there are four edges between the vertices $\gamma^r_i$, $\gamma^{r'}_i$ and the vertices $\alpha^{(r,r')}_t$, $\alpha^{(r',r)}_t$ of $U^X$ which can not all be strong in a $\Lambda$-satisfying STC-labeling if $\{\gamma^r_{\midd^X(x_i)}, \delta_{\down^X(x_i)}\}$ receives strong color~$T^r_i$ and $\{\gamma^{r'}_{\midd^X(x_i)}, \delta_{\down^X(x_i)}\}$ receives strong color $F^{r'}_i$. (Recall that we do not allow weak edges.) To this end, we assign a set of $12$ endpoints in $U^X$ to each variable $x_i$. We need to ensure in particular that two variables $x_i,x_j$ with $\midd^X(x_i)=\midd^X(x_j)$ do not use the same endpoints in $U^X$. We define $\up^X(x_i) := \down^X(x_i)$. The following is directly implied by Claim \ref{Claim: representing edges not equal X}.

\begin{cla} \label{Claim: Different variable ups}
Let $x_i, x_j \in X$ with $x_i \neq x_j$. If $\midd^X(x_i) = \midd^X(x_j)$, then $\up^X(x_i) \neq \up^X(x_i)$.
\end{cla}

We add the following edges between the vertices of $M^X$ and $U^X$: For every variable $x_i$, every $r \in \{1,2,3,4\}$, and every $r' \in \{1,2,3,4\} \setminus \{r\}$ we add the edges~$\{\alpha^{(r,r')}_{\up^X(x_i)} , \gamma^r_{\midd^X(x_i)}\}$, and~$\{\alpha^{(r,r')}_{\up^X(x_i)} , \gamma^{r'}_{\midd^X(x_i)}\}$.

We complete the description of the variable-soundness gadget by defining the vertex lists $\Lambda(v)$ for each $v \in U^X$. We set
\begin{align*}
\Lambda(\alpha^{(r,r')}_t) &:= \bigcup_{\substack{x_i \in X \\ \up^X(x_i)=t}} \{T_i^r, F_i^{r'},R_i, Z_1\} &\text{for every }\alpha^{(r,r')}_t \in U^X \text{.}
\end{align*}

\begin{cla} \label{Claim: Edge Colors Variable Soundness Alpha}
Let $x_i \in X$, let $r \in \{1,2,3,4\}$, and let $r' \in \{1,2,3,4\} \setminus \{r\}$. Then
\begin{enumerate}
\item[a)] $\Lambda(\alpha^{(r,r')}_{\up^X(x_i)}) \cap \Lambda(\gamma^r_{\midd^X(x_i)}) = \{T^r_i, R_i \}$, and
\item[b)] $\Lambda(\alpha^{(r,r')}_{\up^X(x_i)}) \cap \Lambda(\gamma^{r'}_{\midd^X(x_i)}) = \{F^{r'}_i, R_i \}$.
\end{enumerate}
\end{cla}

\begin{claimproof}
We first prove statement a). Let $\Lambda(i,r,r'):= \Lambda(\alpha^{(r,r')}_{\up^X(x_i)}) \cap \Lambda(\gamma^r_{\midd^X(x_i)})$. Clearly, $\Lambda(i,r,r')$ contains~$T_i^r$ and~$R_i$. It remains to show that there is no other strong color $Y \in \Lambda(i,r,r')$. Recall that $$\Lambda(\gamma^r_t) :=  \bigcup_{\substack{x_i \in X \\ \midd^X(x_i)=t}} \{T_i^r, F_i^r, R_i\}\text{.}$$
In the following case distinction we consider every possible strong color $Y \in \Lambda(\gamma^r_{\midd^X(x_i)})$.

\textbf{Case a.1:} $Y= R_j$ or $Y= T^r_j$ for some $j \neq i$. Then, there is a variable $x_j \neq x_i$ with $\midd^X(x_j)=\midd^X(x_i)$. It follows by Claim \ref{Claim: Different variable ups} that $\up^X(x_j) \neq \up^X(x_i)$ and therefore $R_j, T^r_j \not \in \Lambda(\alpha^{(r,r')}_{\up^X(x_i)})$. Hence, $Y \not \in \Lambda(i,r,r')$.

\textbf{Case a.2:} $Y= F^r_j$. Then, since $$\{F^p_t \mid p \in \{1,2,3,4\}, t \in \{1, \dots, n\} \} \cap \Lambda(\alpha^{(r,r')}_{\up^X(x_i)}) \subseteq \{F^{r'}_1, F^{r'}_2, \dots, F^{r'}_n \}\text{ and }r' \neq r$$ we conclude $F^r_j \not \in \Lambda(\alpha^{(r,r')}_{\up^X(x_i)})$. Hence, $Y \not \in \Lambda(i,r,r')$.

Next, we prove statement b) with analogous arguments. Let $\overline{\Lambda}(i,r,r'):= \Lambda(\alpha^{(r,r')}_{\up^X(x_i)}) \cap \Lambda(\gamma^{r'}_{\midd^X(x_i)})$. Clearly, $\{F_i^{r'},R_i\} \subseteq \Lambda(i,r,r')$. It remains to show that there is no other color~$Y \in \overline{\Lambda}(i,r,r')$.

\textbf{Case b.1:} $Y = R_j$ or $Y=F^{r'}_j$ for some $j \neq i$. Then, analogously to Case a.1 we conclude that $Y \not \in \overline{\Lambda}(i,r,r')$.

\textbf{Case b.2:} $Y = T_j^{r'}$. Then, since
$$\{T^p_t \mid p \in \{1,2,3,4\}, t \in \{1, \dots, n\} \} \cap \Lambda(\alpha^{(r,r')}_{\up^X(x_i)}) \subseteq \{T^{r}_1, T^{r}_2, \dots, T^{r}_n \}\text{ and }r' \neq r$$ we conclude $T^{r'}_j \not \in \Lambda(\alpha^{(r,r')}_{\up^X(x_i)})$. Hence, $Y \not \in \overline{\Lambda}(i,r,r')$. This completes the proof of Claim \ref{Claim: Edge Colors Variable Soundness Alpha}. $\hfill \Diamond$
\end{claimproof}

This completes the description of the variable gadget. We continue with the description of the clause gadget.

\proofparagraph{The Clause Gadget.} The clause gadget consists of an upper part and a lower part. Let $U^{\mathcal{C}} := \{ \eta_i \mid i \in \{1, \dots, 12  \lceil \sqrt{n}\, \rceil + 1 \}\}$ be the set of upper vertices and $D^{\mathcal{C}} := \{\theta_i \mid  i \in \{1, \dots, \lceil \sqrt{n}\, \rceil \}\}$ be the set of lower vertices. We add edges such that $U^{\mathcal{C}}$ and $D^{\cC}$ each form cliques in~$G$.

Recall that for some clause $C_j \in \cC$ and a variable $x_i$ occurring in $C_j$ the occurrence number~${\Omega(C_j,x_i)}$ is defined as the number of clauses in $\{C_1,C_2,\dots, C_j\}$ that contain~$x_i$. Below we define two mappings $\up^{\mathcal{C}}: {\mathcal{C}} \rightarrow \{1,2,\dots, 12\sn +1\}$, $\down^{\mathcal{C}}: {\mathcal{C}} \rightarrow \{1,2,\dots,\sn\}$, and vertex lists $\Lambda: V \rightarrow 2^{\{1,\dots,c\}}$. Then, for each clause $C_j \in {\mathcal{C}}$ we add an edge $\{ \eta_{\up^{\cC}(C_j)}, \theta_{\down^\cC (C_j)} \}$. Next, we ensure that this edge can only be labeled with the strong colors that match the literals in~$C_j$. This means, for example, if~$C_j=(x_1 \lor \overline{x_2} \lor x_3)$ we have $\Lambda (\eta_{\up^{\cC}(C_j)})\cap \Lambda(\theta_{\down^\cC (C_j)}) =\{  T^{\Omega(C_j,x_1)}_1, F^{\Omega(C_j,x_2)}_2, T^{\Omega(C_j,x_3)}_3 \}$.

As before, we need to ensure that each variable occurring in a clause has a unique edge between the clause and variable gadgets which transmits the variable's truth assignment to the clause. To achieve this, we define the \emph{clause-conflict graph} $H^ {\mathcal{C}}_{\phi} := ({\mathcal{C}}, \Confl^\mathcal{C})$ by 
\begin{align*}
\Confl^{\mathcal{C}}:= \{ \{ C_i, C_j \} \mid \text{ }&C_i  \text{ contains a variable } x_i \text{ and }
C_j \text{ contains a variable } x_j  \\
&\text{such that } \down^X(x_i)=\down^X(x_j) \} \text{.}
\end{align*}

Clauses that share a variable are one example for adjacent vertices in~$H_\phi^\cC$. Furthermore, due to the compression there may be distinct variables that are mapped to the same value under~$\text{down}^X$. Two distinct clauses containing these variables are another example for adjacent vertices in~$H_\phi^\cC$.

 However, from the fact that each variable occurs in at most four clauses in combination with Claim~\ref{Claim: Groupsizes}~a), it follows that the maximum degree of $H^{\mathcal{C}}_{\phi}$ is at most $12 \cdot \lceil \sqrt{n}\, \rceil$. Thus, there exists a proper vertex coloring~$\chi: {\mathcal{C}} \rightarrow \{1,2, \dots, 12 \cdot \lceil \sqrt{n}\, \rceil +1\}$ such that each color class $\chi^{-1}(i)$, $i \in \{1, \dots, 12 \cdot \lceil \sqrt{n}\, \rceil +1\}$, contains at most~$\lceil \frac{m}{12 \cdot \lceil \sqrt{n}\, \rceil +1} \rceil +1 \leq \lceil \sqrt{n}\, \rceil$ clauses~\cite{HS70}. Such coloring is known as \emph{equitable coloring} and since it has~$\Oh(\sqrt{n})$ colors, it can be computed in polynomial time~\cite{KKKMS10}.

For a given clause $C_i \in \mathcal{C}$ we define $\up^\cC(C_i):=j$ as the index of the color class $\chi^{-1}(j)$ that contains~$C_i$. The following claim provides a useful property for the clause gadget and can be shown with similar arguments as Claim~\ref{Claim: Different variable downs}.

\begin{cla} \label{Claim: Different clause ups}
If a clause $C_{j_1} \in \mathcal{C}$ contains a variable~$x_{i_1}$ and a clause $C_{j_2} \in \mathcal{C}$ contains a variable~$x_{i_2}$ such that $\down^X(x_{i_1})=\down^X(x_{i_2})$, then $\up^{\mathcal{C}}(C_{j_1}) \neq \up^{\mathcal{C}}(C_{j_2})$.
\end{cla}

\begin{claimproof}
By definition, $C_{j_1}$ and $C_{j_2}$ are adjacent in $H^ {\mathcal{C}}_{\phi}$. Hence, $C_{j_1}$ and $C_{j_2}$ are elements of different color classes and therefore $\up^{\mathcal{C}}(C_{j_1}) \neq \up^{\mathcal{C}}(C_{j_2})$. $\hfill \Diamond$
\end{claimproof}

Next, we define $\down^{\mathcal{C}}$ analogously to $\up^X$. To this end consider the finite sequence $\Seq^m= (\up^\cC(C_1), \up^\cC(C_2), \dots, \up^\cC(C_m))$ and define $\down^\cC(C_j)$ as the number of occurrences of $\up^\cC(C_j)$ in the finite sequence $\Seq^j:= (\up^\cC(C_1), \dots, \up^\cC(C_j))$. The fact that each color class contains at most $\lceil \sqrt{n}\, \rceil$ elements implies~$\down^\cC(C_j) \leq \lceil \sqrt{n}\, \rceil$. Intuitively, the following claim guarantees that distinct clauses correspond to distinct edges, which is an analogue statement to Claim \ref{Claim: representing edges not equal X}.

\begin{cla} \label{Claim: clause edges not equal}
Let $C_i, C_j \in \mathcal{C}$. If $C_i \neq C_j$, then $\{ \eta_{\up^{\cC}(C_i)}, \theta_{\down^\cC (C_j)}  \} \neq \{ \eta_{\up^{\cC}(C_j)}, \theta_{\down^\cC (C_j)}  \}$.
\end{cla}
\begin{claimproof}
Without loss of generality, let $i < j$. The claim obviously holds if $\up^\cC (C_i) \neq \up^\cC (C_j)$, so let~$\up^\cC (C_i) = \up^\cC (C_j)$. Then, there is at least one more occurrence of $\up^\cC(C_i)$ in the partial sequence~$\Seq^j_1$ compared to~$\Seq^i_1$. Therefore $\down^\cC (C_i) \neq \down^\cC (C_j)$. $\hfill \Diamond$
\end{claimproof}

We complete the description of the clause gadget by defining the vertex lists $\Lambda(v)$ for every $v \in U^{\mathcal{C}} \cup D^{\mathcal{C}}$. For a given clause $C_j \in  \mathcal{C}$ we define the \emph{color set} $\mathfrak{X}(C_j)$ and the \emph{literal color set} $\mathfrak{L}(C_j)$ of $C_j$ by
\begin{align*}
\mathfrak{X} (C_j) := & \text{ } \{ T^{\Omega(C_j,x_i)}_i, F^{\Omega(C_j,x_i)}_i \mid  x_i \text{ occurs in } C_j \} \text{, and}\\
\mathfrak{L} (C_j) := & \text{ } \{ T^{\Omega(C_j,x_i)}_i \mid  x_i \text{ occurs as a positive literal in } C_j \} \cup {}\\
 &\text{ } \{ F^{\Omega(C_j,x_i)}_i \mid  x_i \text{ occurs as a negative literal in } C_j \} \text{.}
\end{align*}

Note that $\mathfrak{L}(C_j) \subseteq \mathfrak{X}(C_j)$. The vertex lists for the vertices in $U^\cC \cup D^\cC$ are defined as
\begin{align*}
\Lambda(\eta_t) &:=  \bigcup_{\substack{C_j \in \cC \\ \up^\cC(C_j)=t}} \mathfrak{X}(C_j) \cup \{Z_3\} &\text{ for every } \eta_t \in U^\cC \text{, and}\\
\Lambda(\theta_t) &:=  \bigcup_{\substack{C_j \in \cC \\ \down^\cC(C_j)=t}} \mathfrak{L}(C_j) \cup \{Z_4\}  &\text{ for every } \theta_t \in D^\cC \text{.}
\end{align*}

\begin{cla} \label{Claim: Edge Colors Clauses}
Let $C_j \in \cC$. Then, $\Lambda(\eta_{\up^\cC(C_j)}) \cap \Lambda(\theta_{\down^\cC(C_j)}) = \mathfrak{L}(C_j)$.
\end{cla}

\begin{claimproof}
Let $\Lambda(j):= \Lambda(\eta_{\up^\cC(C_j)}) \cap \Lambda(\theta_{\down^\cC(C_j)})$.
Since $\mathfrak{L}(C_j) \subseteq \mathfrak{X}(C_j)$ it holds that~$\mathfrak{L}(C_j) \subseteq \Lambda(j)$. It remains to show that there is no other strong color $Y \in \Lambda(j) \setminus \mathfrak{L}(C_j)$.

\textbf{Case 1:} $Y \in \{Z_3, Z_4\}$. Then, since $Z_3 \not \in \Lambda(\theta_{\down^\cC(C_j)})$ and $Z_4 \not \in \Lambda(\eta_{\up^\cC(C_j)})$, it follows that $Y \not \in \Lambda(j)$.

\textbf{Case 2:} $Y \not \in \{Z_3, Z_4\}$. Assume towards a contradiction that~$Y \in \Lambda(j)$. From~$Y \in \Lambda(\theta_{\down^\cC(C_j)})$ it follows that there is a clause~$C_{j_1}$ with~$\down^\cC(C_{j_1})=\down^\cC(C_j)$ and~$Y \in \mathfrak{L}(C_{j_1})$. It holds that~$C_{j_1} \neq C_j$, since otherwise~$Y \in \mathfrak{L}(C_{j})$, which contradicts the fact that~$Y \in \Lambda(j) \setminus \mathfrak{L}(C_j)$. From~$Y \in \Lambda(\eta_{up^\cC(C_j)})$ it follows that there is a clause~$C_{j_2}$ with~$\up^\cC(C_{j_2})=\up^\cC(C_j)$ and~$Y \in \mathfrak{X}(C_{j_2})$. By the definition of~$\mathfrak{X}$ and~$\mathfrak{L}$ there exists a variable~$x_i$ that occurs in~$C_{j_1}$ and~$C_{j_2}$ such that~$Y=T_i^{\Omega(C_{j_1},x_i)}=T_i^{\Omega(C_{j_2},x_i)}$ or~$Y=F_i^{\Omega(C_{j_1},x_i)}=F_i^{\Omega(C_{j_2},x_i)}$. We conclude~$\Omega(C_{j_1},x_i)=\Omega(C_{j_2},x_i)$ and therefore~$C_{j_2}=C_{j_1}\neq C_j$. Then, the fact that~$\up^\cC(C_{j_1})=\up^\cC(C_j)$ and~$\down^\cC(C_{j_2})=\down^\cC(C_j)$ contradicts Claim~\ref{Claim: clause edges not equal} and therefore~$Y \not \in \Lambda(j)$. $\hfill \Diamond$
\end{claimproof}

\proofparagraph{Connecting the Gadgets.} We complete the construction of $G$ by describing how the vertices of the variable gadget and the vertices of the clause gadget are connected. The idea is to define edges between the vertices in $D^X$ and $U^{\mathcal{C}}$ that model the occurrences of variables in the clauses. 

For each~$C_j \in \cC$ we do the following: Let $x_{i_1}$, $x_{i_2}$, and $x_{i_3}$ be the variables that occur in $C_j$. We add the edges: $\{ \delta_{\down^X(x_{i_1})}, \eta_{\up^\cC(C_j)} \}$, $\{ \delta_{\down^X(x_{i_2})}, \eta_{\up^\cC(C_j)} \}$, and~$\{ \delta_{\down^X(x_{i_3})}, \eta_{\up^\cC(C_j)} \}$. The idea is that an edge~$\{ \delta_{\down^X(x_{i})}, \eta_{\up^\cC(C_j)} \}$ transmits the truth value of a variable~$x_i$ to a clause~$C_j$, where~$x_i$ occurs as a positive or negative literal. The following claim states that the possible strong colors for such an edge are only~$T^{\Omega(C_j,x_i)}_i$ and $F^{\Omega(C_j,x_i)}_i$, which correspond to the negated truth assignment of the~$\Omega(C_j,x_i)$-th occurrence of~$x_i$.

\begin{cla} \label{Claim: Edge Colors Connection}
Let $C_j \in \cC$ be a clause and let $x_i \in X$ be some variable that occurs in $C_j$. Then~$\Lambda(\delta_{\down^X(x_{i})}) \cap \Lambda(\eta_{\up^\cC(C_j)}) = \{T^{\Omega(C_j,x_i)}_i,F^{\Omega(C_j,x_i)}_i \}$.
\end{cla}

\begin{claimproof}
Let $\Lambda(i,j):= \Lambda(\delta_{\down^X(x_{i})}) \cap \Lambda(\eta_{\up^\cC(C_j)})$. Obviously, $\{T_i^{\Omega(C_j,x_i)}, F^{\Omega(C_j,x_i)}_i \} \subseteq \Lambda(i,j)$. It remains to show that there is no strong color $Y \in \Lambda(i,j) \setminus \{T_i^{\Omega(C_j,x_i)}, F^{\Omega(C_j,x_i)}_i \}$.

\textbf{Case 1:} $Y = Z_3$ or $Y= Z_2$. Since $Z_3 \not \in \Lambda(\delta_{\down^X(x_{i})})$ and $Z_2 \not \in \Lambda(\eta_{\up^\cC(C_j)})$, we have~$Y \not \in \Lambda(i,j)$.

\textbf{Case 2:} $Y = T_t^r$ or $Y = F_t^r$ with $t \neq i$ and $r \in \{1,2,3,4\}$. If~$Y \not \in \Lambda(\eta_{\up^\cC(C_j)})$, then obviously~$Y \not \in \Lambda(i,j)$. Thus, let $Y \in \Lambda(\eta_{\up^{\cC}(C_j)})$. Then, by the definition of~$\mathfrak{X}$, there is a clause $C_{j'}$ containing a variable $x_t \neq x_i$ with $\up^\cC(C_{j})= \up^\cC(C_{j'})$. If~$C_{j'}=C_j$, then Claim \ref{Claim: Different variable downs} implies~$\down^X(x_i)\neq\down^X(x_t)$ and thus~$Y \not \in \Lambda(\delta_{\down^X(x_{i})})$. Otherwise, if~$C_{j'} \neq C_j$, then Claim~\ref{Claim: Different clause ups} together with the fact that~$\up^\cC(C_{j})= \up^\cC(C_{j'})$ imply that~$\down^X(x_i)\neq\down^X(x_t)$. Consequently, $Y \not \in \Lambda(\delta_{\down^X(x_{i})}) \supseteq \Lambda(i,j)$.

\textbf{Case 3:} $Y = T_i^r$ or~$Y=F_i^r$ with~$r \neq \Omega(C_j,x_i)$. Obviously, ~$Y \in \Lambda(\delta_{\down^X(x_i)})$. Assume towards a contradiction that~$Y \in \Lambda(\eta_{\up^\cC(C_j)})$. Then, by the definition of the color set~$\mathfrak{X}(\cdot)$, there is a clause $C_{j'}$ containing $x_i$ such that $\up^\cC(C_{j'}) = \up^\cC(C_j)$ and $\Omega(C_{j'},x_i)=r\neq \Omega(C_j,x_i)$. It follows that $C_{j'} \neq C_j$ which contradicts Claim \ref{Claim: Different clause ups}. Hence, $Y \not \in \Lambda(\eta_{\up^\cC(C_j)}) \supseteq \Lambda(i,j)$.
 $\hfill \Diamond$
\end{claimproof}

This completes the description of the construction and basic properties of the \textsc{VL-Multi-STC} instance~$(G,9n+4,0,\Lambda)$. Note that~$G$ has~$\mathcal{O}(\sqrt{n})$ vertices. It remains to show the correctness of the reduction.

\proofparagraph{Correctness.} We show that there is a satisfying assignment for~$\phi$ if and only if there is a~$(9n+4)$-colored~$\Lambda$-satisfying STC-labeling~$L$ for~$G$ with strong color classes 
\begin{align*}
&S_L^{T^r_t},S_L^{F^r_t}, S_L^{R_t}, S_L^{Z_r} &\text{ for all }t \in \{1, \dots, n\} \text{ and } r \in \{1,2,3,4\} \text{,}
\end{align*}
and $W_L = \emptyset$. 

$(\Rightarrow)$ Let $A: X \rightarrow \{\text{`true'},\text{`false'}\}$ be a satisfying assignment for $\phi$. We describe to which strong color classes we add the edges of $G$ so that we obtain a $\Lambda$-satisfying STC-labeling.

First, we describe to which strong color classes we add the edges of the variable gadget. Formally, these are the edges in in~$E(U^X \cup M^X \cup D^X)$. Let $e:=\{\delta_{\down^X(x_i)}, \gamma^r_{\midd^X(x_i)} \}$ be an edge of the variable-representation gadget for some~$x_i \in X$ and~$r \in \{1,2,3,4\}$. We add $e$ to $S_L^{T^r_i}$ if~$A(x_i)=\text{`true'}$ or to $S_L^{F^r_i}$ if $A(x_i)=\text{`false'}$. In both cases, $e$ satisfies the $\Lambda$-list property by Claim~\ref{Claim: Edge Colors Variable Representer}. Next, let $e_1 := \{ \gamma_{\midd^X(x_i)}^r, \alpha_{\up^X(x_i)}^{(r,r')}\}$, and $e_2 := \{ \gamma_{\midd^X(x_i)}^{r'}, \alpha_{\up^X(x_i)}^{(r,r')}\}$ be two edges of the variable-soundness gadget for some $x_i \in X$, $r \in \{1,2,3,4\}$ and $r' \in \{1,2,3,4\} \setminus \{r\}$. We add $e_1$ to $S^{R_i}_L$ if $A(x_i)=\text{`true'}$ or to $S^{T_i^r}_L$ if~$A(x_i)=\text{`false'}$. Further, we add $e_2$ to~$S^{F_i^{r'}}_L$ if $A(x_i)=\text{`true'}$ or to~$S^{R_i}_L$ if $A(x_i)=\text{`false'}$. In each case,~$e_1$ and~$e_2$ satisfy the~$\Lambda$-list property by Claim~\ref{Claim: Edge Colors Variable Soundness Alpha}. For the remaining edges of the variable-gadget we do the following: We add all edges of $E(U^X)$ to $S^{Z_1}_L$ and all edges of $E(D_X)$ to $S^{Z_2}_L$. Obviously, this does not violate the $\Lambda$-list property.

Second, we describe to which strong color classes we add the edges of the clause gadget. Formally, these are the edges in~$E(U^\cC \cup D^\cC)$. Let $C_j \in \cC$ be a clause. Since $A$ satisfies $\phi$, there is some variable $x_i$ occurring in~$C_j$, such that the assignment $A(x_i)$ satisfies the clause $C_j$. Let $r := \Omega(C_j,x_i)$. We add the edge $\{ \eta_{\up^\cC(C_j)}, \theta_{\down^\cC(C_j)} \}$ to $S^{T_i^{r}}_L$ if $A(x_i)= \text{`true'}$ or to $S^{F_i^{r}}_L$ if $A(x_i)=\text{`false'}$. In both cases, the edge satisfies the $\Lambda$-list property by Claim \ref{Claim: Edge Colors Clauses}. For the remaining edges of the clause gadget we do the following: We add all edges of $E(U^\cC)$ to $S^{Z_3}_L$ and all edges of $E(D^\cC)$ to $S^{Z_4}_L$. Obviously, this does not violate the $\Lambda$-list property.

Finally, we describe to which strong color classes we add the edges between the two gadgets. Formally, these are the edges in~$E(D^X,U^\cC)$. Let $C_j \in \cC$ be a clause and let $x_i$ be some variable occurring in $C_j$. Let $r:= \Omega(C_j,x_i)$. We add the edge  $\{ \delta_{\down^X(x_i)}, \eta_{\up^\cC(C_j)} \}$ to $S^{F_i^{r}}_L$ if $A(x_i)=\text{`true'}$ or to $S^{T_i^{r}}_L$ if $A(x_i)= \text{`false'}$. This does not violate the~$\Lambda$-list property by Claim~\ref{Claim: Edge Colors Connection}.

We have now added every edge of $G$ to exactly one strong color class of $L$, such that $L$ is $\Lambda$-satisfying. It remains to show that there is no induced~$P_3$ containing two edges $\{u,v\}$ and $\{v,w\}$ from the same strong color class. In the following case distinction we consider every possible induced~$P_3$ on vertices $u,v$, and~$w$ where $v$ is the central vertex.

\textbf{Case 1:} $ v \in U^X$. Then, $v = \alpha_t^{(r,r')}$ for some $t \in \{1, \dots, \sn + 9\}$, $r \in \{1,2,3,4\}$ and $r' \in \{1,2,3,4\} \setminus \{r\}$. Note that the vertices in $U^X$ are not adjacent to vertices in $D^X$, $U^\cC$ and $D^\cC$. Thus, it suffices to consider the following subcases.

\textbf{Case 1.1:} $u \in U^X$. Then,~$\{u,v\} \in S^{Z_1}_L$. If~$w \in U^X$, then the vertices~$u,v$, and~$w$ do not form an induced~$P_3$, since~$U^X$ is a clique in~$G$. If~$w \not \in U^X$, then~$\{v,w\} \not \in S^{Z_1}_L$. Hence, there is no STC-violation.

%\textbf{Case 1.2:} $u \in U^X$ and $w \in M^X$. Then, $\{u,v\} \in S^{Z_1}_L$ and $\{v,w\} \not \in S^{Z_1}_L$. Thus, there is no STC-violation.

\textbf{Case 1.2:}  $u,w \in M^X$. Then, there are variables~$x_i$ and~$x_j$ with~$\up^X(x_i) = \up^X(x_j) = t$ and~$u= \gamma^p_{\midd^X(x_i)}$,~$w= \gamma^{q}_{\midd^X(x_j)}$ for some~$p,q \in \{r,r'\}$. We need to consider the following subcases.

\textbf{Case 1.2.1:}  $x_i \neq x_j$. Then $i \neq j$. By Claim \ref{Claim: Edge Colors Variable Soundness Alpha} it holds without loss of generality that~$\Lambda(u) \cap \Lambda(v) \subseteq \{T^r_i, F^r_i, T^{r'}_i, F^{r'}_i, R_i \}$ and $\Lambda(w) \cap \Lambda(v) \subseteq \{T^r_j, F^r_j, T^{r'}_j, F^{r'}_j, R_j \}$. Since $L$ is $\Lambda$-satisfying, the edges $\{u,v\}$ and $\{v,w\}$ are elements of different strong color classes. Thus, there is no STC-violation.

\textbf{Case 1.2.2:} $x_i = x_j$. Then, $p \neq q$, since otherwise $u=v$. Without loss of generality, we have~$u = \gamma_{\midd^X(x_i)}^r$ and~$w =  \gamma_{\midd^X(x_i)}^{r'}$. If~$A(x_i)=\text{`true'}$, it follows that~$\{u,v\} \in S^{R_i}_L$ and~$\{v,w\} \in S^{F_i^{r'}}_L$. Otherwise, if~$A(x_i)=\text{`false'}$, it follows that~$\{u,v\} \in S^{T_i^{r'}}_L$ and~$\{v,w\} \in S^{R_i}_L$. In both cases, the edges~$\{u,v\}$ and~$\{v,w\}$ are elements of different strong color classes. Thus, there is no STC-violation.

\textbf{Case 2:} $v \in M^X$. Then $v = \gamma^r_t$ for some $t \in \{1,\dots,\sn\}$ and $r \in \{1,2,3,4\}$. Note that the vertices in $M^X$ are not adjacent to vertices in $U^{\mathcal{C}}$, $D^\cC$ and $M^X$. Thus, it suffices to consider the following subcases.

\textbf{Case 2.1:} $u,w \in U^X$ or $u,w \in D^X$. Then, since $U^X$ and $D^X$ are cliques in $G$, the vertices~$u,v$, and~$w$ do not form an induced~$P_3$ in $G$. Hence, there is no STC-violation.

\textbf{Case 2.2:} $ u \in U^X$ and $w \in D^X$. Then, there are variables~$x_i$ and~$x_j$ with $\midd^X(x_i)=\midd^X(x_j)=t$ and~$u \in \{\alpha^{(r,r')}_{\up^X(x_i)}, \alpha^{(r',r)}_{\up^X(x_i)} \}$,~$w=\delta_{\down^X(x_j)}$ for some~$r' \neq r$. We need to consider the following subcases.

\textbf{Case 2.2.1:} $x_i \neq x_j$. Then, $i \neq j$. Without loss of generality it holds by Claim \ref{Claim: Edge Colors Variable Soundness Alpha} that $\Lambda(u) \cap \Lambda(v) \subseteq \{T^r_i, F^r_i, T^{r'}_i, F^{r'}_i, R_i \}$ for some $r' \neq r$ and by Claim \ref{Claim: Edge Colors Variable Representer} that $\Lambda(v) \cap \Lambda(w) = \{T^r_j,F^r_j\}$. Since $L$ is $\Lambda$-satisfying, the edges $\{u,v\}$ and $\{v,w\}$ are elements of different strong color classes. Thus, there is no STC-violation.

\textbf{Case 2.2.2:}  $x_i = x_j$. Then, if $A(x_i)= \text{`true'}$ it follows that $\{u,v\} \in S_L^{R_i} \cup S_L^{F_i^r}$ and $\{v,w\} \in S_L^{T_i^r}$. If $A(x_i)=\text{`false'}$ it follows that $\{u,v\} \in S_L^{R_i} \cup S_L^{T_i^r}$ and $\{v,w\} \in S_L^{F_i^r}$. In both cases the edges $\{u,v\}$ and $\{v,w\}$ are elements of different strong color classes. Thus, there is no STC-violation.

\textbf{Case 3:} $v \in D^X$. Then $v = \delta_t$ for some $t \in \{1,\dots, \sn + 9 \}$. Note that the vertices in $D^X$ are not adjacent to vertices in $U^X$ and $D^\cC$. Thus, it suffices to consider the following subcases.

\textbf{Case 3.1:} $u \in D^X$. Then, $\{u,v\} \in S^{Z_2}_L$. If $w \in D^X$, then the vertices $u,v$, and~$w$ do not form an induced~$P_3$, since $D^X$ is a clique in $G$. If $w \not \in D^X$, then~$\{v, w\} \not \in S^{Z_2}_L$. Hence, there is no STC-violation.

\textbf{Case 3.2:} $u,w \in U^{\mathcal{C}}$. Then, the vertices $u,v$, and~$w$ do not form an induced $P_3$, since $U^\cC$ forms a clique.

\textbf{Case 3.3:} $u, w \in M^X$. By Claim~\ref{Claim: Edge Colors Variable Representer}, all edges $\{v,y\} \in E(\{v\},M^X)$ have distinct possible strong colors in~$\Lambda(v) \cap \Lambda(y)$. Since $L$ is $\Lambda$-satisfying, the edges~$\{u,v\}$ and~$\{v,w\}$ are elements of different strong color~classes.

 %Then, there are variables $x_i, x_j$, with $\text{down}^X(x_i)= \text{down}^X(x_j)=t$. If $x_i = x_j$ it follows that $u= \gamma_{\midd^X(x_i)}^r$ and $w= \gamma_{\midd^X(x_i)}^{r'}$ with $r' \neq r$. Then, by claim \ref{Claim: Edge Colors Variable Representer} and the fact that $L$ is $\Lambda$-satisfying, $\{u,v\}$ and $\{v,w\}$ are not elements of the same strong color class. Thus, there is no STC-violation.

\textbf{Case 3.4:} $u \in M^X$ and $w \in U^\cC$. Then,~$u= \gamma_{\midd^X(x_i)}^{r}$ for some $x_i \in X$ with $\down^X(x_i)=t$ and $r \in \{1,2,3,4\}$. Moreover,~$w = \eta_{\up^\cC(C_j)}$ for some clause~$C_j$ containing a variable $x_{i'}$ with~$\down^X(x_{i'})=t$. We need to consider the following subcases.

\textbf{Case 3.4.1:} $x_i \neq x_{i'}$. Then,~$i \neq i'$ and by Claim~\ref{Claim: Edge Colors Variable Representer} we have~$\Lambda(u) \cap \Lambda(v) = \{T^r_i, F^r_i\}$ and by Claim \ref{Claim: Edge Colors Connection} we have~$\Lambda(v) \cap \Lambda(w) = \{T^{r'}_{i'},F^{r'}_{i'}\}$ with~$r' = \Omega(C_j,x_{i'})$. Then, since~$L$ is~$\Lambda$-satisfying,~$\{u,v\}$ and~$\{v,w\}$ are not elements of the same strong color class.

\textbf{Case 3.4.2:} $x_i = x_{i'}$. Then, if $A(x_i)= \text{`true'}$ it follows that~$\{u, v \} \in S_L^{T_i^r}$ and~$\{v,w\} \in S_L^{F^{r'}_i}$ for some $r' \in \{1,2,3,4\}$. If $A(x_i)= \text{`false'}$, then it follows that~$\{u, v \} \in S_L^{F_i^r}$ and $\{v,w\} \in S_L^{T^{r'}_i}$ for some~$r' \in \{1,2,3,4\}$. In both cases $\{u,v\}$ and $\{v,w\}$ are elements of different strong color classes.

\textbf{Case 4:} $v \in U^\cC$. Then, $v= \eta_{t}$ for some $t \in \{1,\dots,12\sn +1\}$. Note that the vertices in $U^\cC$ are not adjacent to vertices in $U^X$ and $M^X$. Thus, it suffices to consider the following subcases.

\textbf{Case 4.1:} $u \in U^\cC$. Then, $\{u,v\} \in S^{Z_3}_L$. If $w \in U^{\cC}$, then the vertices $u,v$, and~$w$ do not form an induced $P_3$ since $U^X$ is a clique in $G$. If $w \not \in U^\cC$ it follows that~$\{v,w\} \not \in S^{Z_3}_L$. Hence, there is no STC-violation.

\textbf{Case 4.2:} $u,w \in D^X$ or $u,w \in D^\cC$. Then, the vertices $u,v$, and~$w$ do not form an induced $P_3$, since $D^X$ and $D^\cC$ form cliques in $G$.

\textbf{Case 4.3:} $u \in D^X$ and $w \in D^\cC$. Then, there is a clause $C_j$ with $\up^\cC(C_j)=t$ and a clause~$C_{j'}$ containing a variable $x_i$ with $\up^\cC(C_{j'})=t$ and $u=\delta_{\down^X(x_i)}$, $w=\theta_{\down^\cC(C_j)}$. We consider the following subcases.

\textbf{Case 4.3.1:} $C_j \neq C_{j'}$. Then, since $\up^\cC(C_j)=\up^\cC(C_{j'})$ it follows by Claim~\ref{Claim: Different clause ups} that~$C_j$ and~$C_{j'}$ do not share a variable. Hence,~$x_i$ does not occur in~$C_j$ and therefore~$T^{\Omega(C_{j'},x_i)}_i, F^{\Omega(C_{j'},x_i)}_i \not \in \mathfrak{L}(C_j)$. Thus, by Claims~\ref{Claim: Edge Colors Clauses} and~\ref{Claim: Edge Colors Connection} and the fact that~$L$ is~$\Lambda$-satisfying, the edges~$\{u,v\}$ and~$\{v,w\}$ are elements of different strong color classes.

\textbf{Case 4.3.2:} $C_j = C_{j'}$. Let $r:= \Omega(C_j,x_i)$. If $\{v,w\} \not \in S^{T^r_i}_L \cup S^{F^r_i}_L$, the edges $\{u,v\}$ and $\{v,w\}$ are elements of different color classes. Thus, there is no STC-violation. If $\{v,w\} \in S^{T^r_i}_L \cup S^{F^r_i}_L$ it follows by the construction of $L$ that $C_j$ is satisfied by the assignment $A(x_i)$. Without loss of generality assume that $x_i$ occurs as a positive literal in $C_j$. Then, $A(x_i)= \text{`true'}$. This implies~$\{v,w\} \in S^{T^r_i}_L$ and~$\{u,v \} \in S_L^{F_i^{r}}$. Hence, $\{u,v\}$ and $\{v,w\}$ are elements of different strong color classes.

\textbf{Case 5:} $v \in D^\cC$. Then, $v$ is not adjacent with any vertices in $U^X$, $M^X$ or~$D^X$. Hence, we need to consider the following cases.

\textbf{Case 5.1:} $u,w \in U^\cC$ or $u,w \in D^\cC$. Then, the vertices $u,v$, and~$w$ do not form an induced $P_3$ since $U^\cC$ and $D^\cC$ are cliques in $G$. 

\textbf{Case 5.2:} $u \in D^\cC$ and $w \in U^\cC$. Then, $\{u,v\} \in S^{Z_4}_L$ and $\{v,w\} \not \in S^{Z_4}_L$. Hence, there is no STC-violation.

This proves that $L$ is a $\Lambda$-satisfying STC-labeling for $G$ with no weak edges.

$(\Leftarrow)$ Conversely, let $L$ be a $(9n+4)$-colored $\Lambda$-satisfying STC-labeling for $G$. We show that $\phi$ is satisfiable. We define an assignment $A: C \rightarrow \{\text{`true'}, \text{`false'}\}$ by
\begin{align*}
A(x_i) := \begin{cases}
\text{`true'} & \text{if } \{\delta_{\down^X(x_i)}, \gamma^1_{\midd^X(x_i)}\} \in S^{T^1_i}_L \text{, and}\\
\text{`false'} & \text{if } \{\delta_{\down^X(x_i)}, \gamma^1_{\midd^X(x_i)}\} \in S^{F^1_i}_L \text{.}
\end{cases}
\end{align*}
The assignment is well-defined due to Claim \ref{Claim: Edge Colors Variable Representer}. The following claim states that, if one occurrence~$r \in \{1,2,3,4\}$ of some variable~$x_i$ that is assigned `true' (or `false', respectively), then so is the first occurrence of~$x_i$. We obtain this statement by using the variable-soundness gadget.

\begin{cla} \label{Claim: Correctness Soundness}
Let $x_i \in X$ and $r \in \{2,3,4\}$.
\begin{enumerate}
\item[a)] If $\{\delta_{\down^X(x_i)}, \gamma^r_{\midd^X(x_i)}\} \in S^{T^r_i}_L$, then $\{\delta_{\down^X(x_i)}, \gamma^1_{\midd^X(x_i)}\} \in S^{T^1_i}_L$.
\item[b)] If $\{\delta_{\down^X(x_i)}, \gamma^r_{\midd^X(x_i)}\} \in S^{F^r_i}_L$, then $\{\delta_{\down^X(x_i)}, \gamma^1_{\midd^X(x_i)}\} \in S^{F^1_i}_L$.
\end{enumerate}
\end{cla}
\begin{claimproof}
We show (a). Let $\{\delta_{\down^X(x_i)}, \gamma^r_{\midd^X(x_i)}\} \in S^{T^r_i}_L$. Consider the vertex $\alpha^{(r,1)}_{\up^X(x_i)}$. By Claim \ref{Claim: Edge Colors Variable Soundness Alpha} we have 
\begin{align*}
\Lambda(\alpha^{(r,1)}_{\up^X(x_i)}) \cap \Lambda(\gamma^r_{\midd^X(x_i)}) &= \{T^r_i,R_i\} \text{, and} \\
\Lambda(\alpha^{(r,1)}_{\up^X(x_i)}) \cap \Lambda(\gamma^1_{\midd^X(x_i)}) &= \{F^1_i,R_i\} \text{.}
\end{align*}
Note, that the vertices $\delta_{\down^X(x_i)}, \gamma^r_{\midd^X(x_i)}, \alpha^{(r,1)}_{\up^X(x_i)}$ form an induced~$P_3$ in $G$. Since~$L$ is a $\Lambda$-satisfying STC-labeling with no weak edges, it holds that~$\{\gamma^r_{\midd^X(x_i)}, \alpha^{(r,1)}_{\up^X(x_i)}\} \in S^{R_i}_L$. Then, since the vertices $\gamma^r_{\midd^X(x_i)}$, $\alpha^{(r,1)}_{\up^X(x_i)}$, and~$\gamma^1_{\midd^X(x_i)}$ form an induced~$P_3$, the same argument implies~$\{\alpha^{(r,1)}_{\up^X(x_i)}, \gamma^1_{\midd^X(x_i)} \} \in S^{F^1_i}_L$. Then, since $\Lambda(\delta_{\down^X(x_i)}) \cap \Lambda(\gamma^1_{\midd^X(x_i)}) = \{T^1_i,F^1_i\}$ by Claim \ref{Claim: Edge Colors Variable Representer} and the fact that $\delta_{\down^X(x_i)}$, \(\gamma^1_{\midd^X(x_i)}\), \(\alpha^{(r,1)}_{\up^X(x_i)}\) form an induced~$P_3$ it follows that~$\{\delta_{\down^X(x_i)}, \gamma^1_{\midd^X(x_i)}\} \in S^{T^1_i}_L$ as claimed.

Statement (b) can be shown with the same arguments by considering the vertex $\alpha^{(1,r)}_{\up^X(x_i)}$ instead~of~$\alpha^{(r,1)}_{\up^X(x_i)}$. $\hfill \Diamond$
\end{claimproof}

Next we use Claim \ref{Claim: Correctness Soundness} to show that every clause is satisfied by $A$. Let $C_j \in \cC$ be a clause. Then, there is an edge $e_1:=\{\eta_{\up^\cC(C_j)},\theta_{\down^\cC(C_j)}\} \in E$. By Claim~\ref{Claim: Edge Colors Clauses} we have $\Lambda(\eta_{\up^\cC(C_j)}) \cap \Lambda (\theta_{\down^\cC(C_j)}) = \mathfrak{L}(C_j)$. Since $L$ is $\Lambda$-satisfying, it follows that~$e_1 \in S^Y_L$ for some $Y \in \mathfrak{L}(C_j)$.

Consider the case $Y=T^r_i$ for some variable $x_i$ that occurs positively in $C_j$ and $r= \Omega(C_j,x_i)$. We show that $A(x_i)=\text{`true'}$. Since $x_i$ occurs in $C_j$ there is an edge $e_2:=\{\delta_{\down^X(x_i)},\eta_{\up^\cC(C_j)}\} \in E$ which can only be an element of the strong classes $S_L^{T^r_i}$ or $S_L^{F^r_i}$ due to Claim \ref{Claim: Edge Colors Connection}. Since $e_1$ and $e_2$ form an induced~$P_3$ and $L$ is an STC-labeling we have~$e_2 \in S_L^{F^r_i}$. The edge $e_3:=\{ \delta_{\down^X(x_i)},\gamma^r_{\midd^X(x_i)} \}$ forms an induced~$P_3$ with $e_2$ and can only be an element of the strong classes $S_L^{T^r_i}$ or $S_L^{F^r_i}$ by Claim \ref{Claim: Edge Colors Variable Representer}. Hence, $e_3 \in S_L^{T^r_i}$. By Claim \ref{Claim: Correctness Soundness} we may conclude $\{\delta_{\down^X(x_i)}, \gamma^1_{\midd^X(x_i)}\} \in S^{T^1_i}_L$ and therefore $A(x_i)=\text{`true'}$. Hence, $C_j$ is satisfied by $A$.

For the case $Y=F^r_i$ we can use the same arguments to conclude $A(x_i)=\text{`false'}$. Hence, $A$ satisfies every clause of $\phi$.
\end{proof}

Note that in the instance constructed in the proof of Theorem~\ref{Theorem: ETH lower bound}, every edge has at most three possible strong colors and~$c \in \Oh(n)$ where~$n$ is the number of variables in~$\phi$. This implies the following.

\begin{corollary}
If the ETH is true, then
\begin{enumerate}
\item[a)] \textsc{EL-Multi-STC} cannot be solved in~$2^{o(|V|^2)}$ time even if restricted to instances~$(G,c,k,\Psi)$ where~$k=0$ and~$\max_{e \in E}{|\Psi(e)|}=3$.
\item[b)] \textsc{VL-Multi-STC} cannot be solved in~$c^{o(|V|^2/\log c)}$ time even if\iflong{} restricted to instances where\fi~$k=0$.
\end{enumerate}
\end{corollary}

Since the lower bound holds for instances with~$k=0$, we also obtain a lower bound for approximation algorithms for \textsc{VL-Multi-STC} and \textsc{EL-Multi-STC}.

\begin{corollary}
If the ETH is true, then there exists no approximation algorithm for \textsc{VL-Multi-STC} that runs in time~$2^{o(n^2)}$.
\end{corollary}

\section{Parameterized Complexity}
\label{sec:pc}
Motivated by the strong hardness results from the previous section, we study the parameterized complexity of~\textsc{Multi-STC} and its list variants. The most natural parameter % to consider
is the number~$k$ of weak edges. The case~$c=1$ (STC) is fixed-parameter tractable~\cite{ST14}. For~$c=2$, we also obtain an FPT algorithm: 
An STC-labeling with two strong colors corresponds to a proper two-coloring of the Gallai graph~$\tilde{G}$ of the input graph after deleting the vertices corresponding to weak edges.
Thus for \(c = 2\) \textsc{Multi-STC} reduces to \textsc{Odd Cycle Transversal} in~$\tilde{G}$ which is fixed-parameter tractable with respect to~$k$~\cite{RSV04}.
This extends to \textsc{EL-Multi-STC} with~$c=2$ by modifying the reduction slightly: After computing the Gallai graph, we add two vertex sets~$A$ and~$B$, each of size~$k+1$. Moreover, we add edges such that we obtain a biclique where~$A$ and~$B$ are the partite sets. We then connect each vertex that may only choose color~$1$ with every vertex of~$A$ and connect each vertex that may only choose color~$2$ with every vertex of~$B$.

In contrast, for every fixed $c \geq 3$, \textsc{Multi-STC} is NP-hard even if $k=0$~\cite{LG83}. Thus, FPT algorithms for\iflong\ the parameters $c$, $k$, or even the combined parameter\fi\ $(c,k)$ are unlikely. Instead, we study a structural parameter~$\ko$ that is related to~$k$. Informally,~$\ko$ is the solution size in the case where~$c=1$. Formally, the parameter is defined as follows.

\begin{definition}
Let $G=(V,E)$ be a graph with a $1$-colored STC-labeling $L=(S_L,W_L)$ \iflong such that there is no $1$-colored STC-labeling $L'=(S_{L'},W_{L'})$ for $G$ with $|W_{L'}| < |W_L|$.\else with a minimal number of weak edges.\fi{} Then we set $\ko = \ko(G) := |W_L|$.
\end{definition}
Equivalently, $\ko(G)$ is the size of a minimum vertex cover of the Gallai graph $\tilde{G}$ of~$G$~\cite{ST14}.

\subsection{A Fixed-Parameter Algorithm for $\mathbf{(c,\ko)}$}
We provide a simple FPT algorithm for \textsc{EL-Multi-STC} parameterized by~$(c,\ko)$\iflong{, which is the most general of the three problems\fi. The main idea of the algorithm is to solve \textsc{List-Colorable Subgraph} on the Gallai graph \iflong of the input graph which is equivalent due to Proposition~\ref{Proposition: Gallai equivalence}.\fi 

Let~$(G,c,k,\Psi)$ be an instance of \textsc{EL-Multi-STC}. The first step is to compute the Gallai graph~$\tilde{G}=(\tilde{V},\tilde{E})$ of~$G$ which has~$m$ vertices and at most~$nm$ edges.
  Observe that $(G,c,k,\Psi)$ is equivalent to the instance $(\tilde{G},c,k,\Psi)$ of \textsc{List-Colorable Subgraph} due to Proposition~\ref{Proposition: Gallai equivalence}.
  We describe an algorithm that solves~$(\tilde{G},c,k,\Psi)$ in~$\mathcal{O}((c+1)^{s} \cdot (|\tilde{V}|\cdot c + |\tilde{E}|))$ time, where~$s=\ko$ denotes the size of a minimum vertex cover of~\(\tilde{G}\).

 Let $S \subseteq \tilde{V}$ be a size-$s$ vertex cover of $\tilde{G}$, which can be computed in $\mathcal{O}(1.28^s+sn)$ time~\cite{CKX10}.
 Let $I:=\tilde{V} \setminus S$ denote the remaining independent set.
 We now compute whether $\tilde{G}$ has a subgraph-$c$-coloring~$a: \tilde{V} \rightarrow \{0,1,\dots, c\}$ with~$|\{ v \in \tilde{V} \mid a(v)=0\}| \leq k$.

 We enumerate all possible mappings $a_S:S \rightarrow \{0,1,\dots,c\}$.
 Observe that there are $(c+1)^s$ such mappings.
 For each $a_S$ we check whether $a_S(v) \in \Psi(v) \cup \{0\}$ for all $v \in S$.
 Furthermore, we check in $\mathcal{O}(|\tilde{V}| \cdot c + |\tilde{E}|)$~time whether $a_S$ is a subgraph-$c$-coloring for~$\tilde{G}[S]$.
 If this is not the case, then discard the current $a_S$.
 For every other choice of~$a_S$ we proceed as follows:

 We check whether it is possible to extend $a_S$ to a mapping $a: \tilde{V} \rightarrow \{0,1,\dots,c\}$ that is a subgraph-$c$-coloring for~$\tilde{G}$:
 For each vertex~$v \in I$ we check whether~$P_v := \Psi(v) \setminus \{a_S(w) \mid w \in N_{\tilde{G}}(v)\}$ is empty. If $P_v=\emptyset$ we set $a(v)=0$. Otherwise, we set $a(v)=p$ for some arbitrary $p \in P_v$. This can be done in $\mathcal{O}(|\tilde{V}| \cdot c + |\tilde{E}|)$~time. The resulting mapping $a:V \rightarrow \{0, 1, \dots, c\}$ is a subgraph-$c$-coloring for~$\tilde{G}$, since $a_S$ is a subgraph-$c$-coloring for $G[S]$ and every $v \in I$ has a color $a(v)$ distinct from all vertices in $N(v) \subseteq S$.
 Finally, we check whether the total number of vertices with $a(v)=0$ is at most~$k$.
 If so, then we accept and stop.
 Otherwise, we continue with the next mapping \(a_S\).
 If we did not accept for any of the enumerated mappings \(a_S\), then we reject.
 
  The overall running time of the algorithm is $\mathcal{O}((c+1)^{s} \cdot (nc + m))$.
 Recall that $\ko=s$, $|\tilde{V}|=m$ and $|\tilde{E}| \leq nm$.
 Therefore, we can solve \textsc{EL-Multi-STC} in $\mathcal{O}((c+1)^{\ko} \cdot (cm + nm))$ time.
 
 \begin{lemma}
 The algorithm described above is correct.
 \end{lemma}
 
\begin{proof}

 To see that the above algorithm is correct, observe first that it only accepts if it has found a subgraph-\(c\)-coloring for \(\tilde{G}\) with $|\{ v \in \tilde{V} \mid a(v)=0\}| \leq k$.
 For the other direction, assume that there is a subgraph-\(c\)-coloring~\(a^\star \colon \tilde{V} \to \{0,1,\dots, c\}\) with~$|\{ v \in \tilde{V} \mid a^\star(v)=0\}| \leq k$.
 One of the enumerated mappings \(a_S\) satisfies \(a_S = a^\star|_S\).
 For this mapping \(a_S\) the algorithm sets \(a(v) = 0\) for some vertex \(v \in I\) if and only if \(\Psi(v) \setminus \{a^\star(w) \mid w \in N_{\tilde{G}}(v)\}\) is empty.
 Thus, the number of vertices \(v \in \tilde{V}\) with \(a(v) = 0\) is at most $|\{ v \in \tilde{V} \mid a^\star(v)=0\}| \leq k$, as required.
\end{proof}

We obtain the following fixed-parameter tractability result.

\begin{restatable}{theorem}{elmultifpt}
%\begin{theorem}
\label{Theorem: (c+1)^ko Algorithm for EL-M-STC}
\textsc{EL-Multi-STC} can be solved in $\mathcal{O}((c+1)^{\ko} \cdot (cm + nm))$~time.
%\end{theorem}
\end{restatable}
% \subsection{Multi-STC parameterized by \boldmath$k_1$}
% In this subsection
Next, we conclude that \textsc{Multi-STC} parameterized by $\ko$ alone is fixed-parameter tractable. To this end we observe the following relationship between $c$~and~$\ko$.

\begin{restatable}{lemma}{cvsko}
%\begin{proposition}
\label{Lemma: c leq ko Multi-STC}
Let $G=(V,E)$ be a graph. For all $c, k \in \mathds{N}$ with $c > \ko(G)$ it holds that $(G,c,k)$ is a yes-instance for \textsc{Multi-STC}.
%\end{proposition}
\end{restatable}

\begin{proof}
Let $c> \ko$. Then there exists an STC-labeling $L=(S_L,W_L)$ for~$G$ with one strong color and $|W_L|=\ko$. Let $e_1,e_2, \dots, e_{\ko}$ be the weak edges of~$L$. We define a $c$-colored labeling $L^+ := (S^1_{L^+}, \dots, S^c_{L^+}, W_{L^+} )$ by
\begin{align*}
W_{L^+} := \emptyset \text{ and } S^i_{L^+} := \begin{cases}
      		\{e_i\} & \text{for } i \in \{1 ,\dots,\ko\} \text{,} \\
      		S_L    & \text{for } i = \ko +1 \text{,}\\
      		\emptyset & \text{for } i \in \{k+2, \dots, c \}  \text{.}
    	\end{cases}
\end{align*}
\looseness=-1
Since $c > \ko$, every edge of $G$ is labeled by $L^+$. Because $L$ is an STC-labeling, there is no induced~$P_3$ containing two edges from $S^{\ko+1}_{L^+} = S_L$. Moreover, since $|S^i_{L^+}| \leq 1$ for $i \neq \ko+1$, the labeling~$L^+$ satisfies STC. Since $|W_{L^+}| = 0$, it holds that $(G,c,k)$ is a yes-instance for \textsc{Multi-STC} for every~$k \in \mathds{N}$.
\end{proof}

Lemma~\ref{Lemma: c leq ko Multi-STC} implies an FPT algorithm for \textsc{Multi-STC} parameterized by~$\ko$ alone: Let $(G,c,k)$ be an instance of \textsc{Multi-STC}. If $c > \ko$, then $(G,c,k)$ is a yes-instance by Lemma \ref{Lemma: c leq ko Multi-STC}. We thus only need to consider instances with $c \leq \ko$. Replacing~$c$ by~$k_1$ in the running time bound of Theorem \ref{Theorem: (c+1)^ko Algorithm for EL-M-STC} then gives the following.

\begin{restatable}{corollary}{fptmultiko}
%\begin{theorem} 
\label{Theorem: Multi-STC FPT ko}
\textsc{Multi-STC} can be solved in $\mathcal{O}((\ko+1)^{\ko} \cdot (\ko m + nm))$ time.
%\end{theorem}
\end{restatable}

\subsection{W[1]-hardness for Vertex-List Multi-STC parameterized by $\mathbf{\ko}$}
The fixed-parameter tractability of~\textsc{Multi-STC} parameterized by~$\ko$ alone relies on the relationship between~$c$ and~$\ko$ from Lemma~\ref{Lemma: c leq ko Multi-STC}. Unfortunately, Lemma~\ref{Lemma: c leq ko Multi-STC} does not hold for~\textsc{VL-Multi-STC}: Figure~\ref{Figure: c>k_1 with vertex lists} shows an example.

\begin{figure} 
\begin{center}
\begin{tikzpicture}%[scale=0.85,yscale=0.7]
\tikzstyle{knoten}=[circle,fill=white,draw=black,minimum size=5pt,inner sep=0pt]
\tikzstyle{bez}=[inner sep=0pt]

\node[knoten] (u)  at (0,0) {};
\node[bez] at (0,0.5) {$\{1\}$};
\node[knoten] (v)  at (2,0) {};
\node[bez] at (2,0.5) {$\{1,2\}$};
\node[knoten] (w)  at (4,0) {};
\node[bez] at (4,0.5) {$\{1,2 \}$};
\node[knoten] (x)  at (6,0) {};
\node[bez] at (6,0.5) {$\{2 \}$};

\draw[-, line width=1pt]  (u) to (v);

\draw[-, line width=1pt]  (w) to (v);

\draw[-, line width=1pt]  (w) to (x);

\end{tikzpicture}
\end{center}
\caption{A graph~$G$ and vertex lists~$\Lambda$ of an instance~$I:=(G,c=2,k=0,\Lambda)$ of \textsc{VL-Multi-STC}. It is easy to see that~$I$ is a no-instance while~$k_1=1<2=c$.}
\label{Figure: c>k_1 with vertex lists}
\end{figure}

We now show that there is little hope to obtain fixed-parameter tractability for the list variants of~\textsc{Multi-STC} parameterized by~$\ko$ by giving a parameterized reduction from \textsc{Set Cover} which is defined as follows.
\begin{quote}
  \textsc{Set Cover}\\
  \textbf{Input}: A finite universe $U \subseteq \mathds{N}$, a family $\mathcal{F} \subseteq 2^U$, and an integer $t \in \mathds{N}$.\\
  \textbf{Question}: Is there a subfamily $\mathcal{F}' \subseteq \mathcal{F}$ with $|\mathcal{F}'| \leq t$ such that~$\bigcup_{F \in \mathcal{F}'} F = U$?
\end{quote}
More precisely, we reduce from \textsc{Set Cover} parameterized by the dual parameter~$|\mathcal{F}|-t$.
The W[1]-hardness of  \textsc{Set Cover} for this parameterization  follows from a classic reduction from \textsc{Independent Set} \cite{Karp72}. We provide it here for the sake of completeness.

\begin{restatable}{proposition}{setcoverwone}
% \begin{proposition}
  \label{SetCover-W1}
\textsc{Set Cover} parameterized by $|\mathcal{F}|-t$ is W[1]-hard.
% \end{proposition}
\end{restatable}

\begin{proof}
The \textsc{Independent Set} problem asks for a given graph $G=(V,E)$ and integer~\(s\) whether there is a subset $V' \subseteq V$ of size at least $s$ such that the vertices in $V'$ are pairwise non-adjacent in~$G$. It is known to be W[1]-hard when parameterized by~$s$~\cite{DF13}.

Let $(G=(V,E),s)$ be an instance of \textsc{Independent Set}. We construct a \textsc{Set Cover}-instance $(U,\mathcal{F},t)$ as follows. Set $U:=E$, $\mathcal{F}:=\{  F_v \mid v \in V \}$ with $F_v := \{\{v,u\} \mid u \in N(v)\}$ and $t:=|V|-s$. Note that $|\mathcal{F}|=|V|$, hence $|\mathcal{F}|-t=|V|-(|V|-s)=s$.
\end{proof}

\begin{restatable}{theorem}{vlwone}
  % \begin{theorem}
    \label{VertList-W1-by-ko} \textsc{VL-Multi-STC} parameterized by $\ko$ is
    W[1]-hard, even if~$k=0$.
  % \end{theorem}
\end{restatable}

\begin{proof}
  We give a parameterized reduction from \textsc{Set Cover} parameterized by $|\mathcal{F}|-t$ which is W[1]-hard due to Proposition \ref{SetCover-W1}. For a given \textsc{Set Cover}-instance $(U, \mathcal{F}, t)$ we describe how to construct an equivalent \textsc{VL-Multi-STC}-instance $(G=(V,E),c,k,\Lambda)$ with $\ko \leq |\mathcal{F}| -t$ and $k=0$. Let~$\mathcal{F}=\{ F_1, \dots, F_{|\mathcal{F}|} \}$. Define the vertex set~$V$ of the input graph~$G$ by $V := U \cup Z \cup \{ a \}$ where \(a\) is a new vertex and $Z := \{z_i \mid i \in \{t+1, t + 2, \ldots, |\mathcal{F}| \} \}$.
  Define the edge set of $G$ by $E := E_U \cup E_{Ua} \cup E_{Za}$ with
  \begin{itemize}
  \item $E_{U} := \{\{v, w\} \mid v,w \in U \}$, 
  \item $E_{Ua} := \{ \{u, a\} \mid u \in U \}$, and
  \item     $E_{Za} := \{ \{z,a\} \mid z \in Z \}$.
\end{itemize}
Note that~$|E_{Za}|=|Z|=|\mathcal{F}|-t$ and that~$U$ is a clique.

We let $c := |\mathcal{F}|+1$ and define the lists $\Lambda$ as
\begin{align*}
\Lambda(v) := \begin{cases}
\{ i \mid v \in F_i \} \cup \{|\mathcal{F}|+1\}  & \text{if } v \in U \text{,}\\
\{ 1,2, \dots, |\mathcal{F}|\} & \text{if } v \not \in U \text{.}
\end{cases}
\end{align*}

The idea behind this construction is that the vertex $a$ selects sets from $\mathcal{F}$ by labeling the edges in $E_{Ua}$. The edges in $E_{Za}$ ensure, that there are exactly $t$ different strong colors left for the edges in $E_{Ua}$.

We first show that $\ko \leq |\mathcal{F}|-t$. Let $e_1, e_2 \in E$ be the edges of an induced $P_3$ in $G$. Since $U \cup \{a \}$ is a clique by construction, at least one of the edges $e_1$ or $e_2$ has one endpoint in $Z$, hence it belongs to $E_{Za}$. Since every $P_3$ in $G$ contains at least one edge from $E_{Za}$ it follows that defining $E_{Za}$ as weak edges and $E_U \cup E_{Ua}$ as strong edges yields an STC-labeling with one strong color. This labeling has~$|E_{Za}|=|\mathcal{F}|-t$ weak edges, hence~$\ko \leq |\mathcal{F}|-t$.

It remains to show that $(U, \mathcal{F}, t)$ has a solution $\mathcal{F}'$ of size $t$ if and only if $G$ has a $\Lambda$-satisfying STC-labeling $L=(S^1_L, \dots, S^{|\mathcal{F}|+1}_L, W_L)$ with~$W_L = \emptyset$.

$(\Rightarrow)$ Let $\mathcal{F}' \subseteq \mathcal{F}$ be a set cover of $U$ with $|\mathcal{F}'| = t$. Without loss of generality  (by renaming) let $\mathcal{F}' = \{ F_1, F_2, \dots, F_t\}$. We define an STC-labeling $L=(S^1_L, \dots, S^{|\mathcal{F}|+1}_L, \emptyset)$ as follows.
We start by defining the classes $S^i_L$ for each $i \in \{ t+1, t+ 2, \dots, |\mathcal{F}|+1\}$. We set
\begin{align*}
S^{i}_L := \{ \{a, z_i \} \} \text{ for every } i \in \{t+1, t+2, \ldots, |\mathcal{F}|\} \text{ and } S^{|\mathcal{F}|+1}_L := E_U \text{.}
\end{align*}
Note that $S^{t+1}_L \cup \dots \cup S^{|\mathcal{F}|+1}_L = E_U \cup E_{Za}$, so by defining the strong color classes $S^{t+1}_L, \dots, S^{|\mathcal{F}|+1}_L$ we have labeled all edges in $E_U \cup E_{Za}$. Before we continue with the definition of~$L$, we show that the definition of the strong classes~$S^{t+1}_L, \dots,S^{|\mathcal{F}|+1}_L$ does not violate the STC property and every edge in $E_U \cup E_{Za}$ satisfies the $\Lambda$-list property. Since $U$ is a clique by construction, there is no induced~$P_3$ containing two edges from $S^{|\mathcal{F}|+1}_L$ violating STC in $E_U$. Moreover, since all sets $S^{t+1}_L, \dots, S^{|\mathcal{F}|}$ contain exactly one edge, there is obviously no STC violation in $E_{Za}$. For every vertex $u \in U$ it holds that $|\mathcal{F}|+1 \in \Lambda(u)$, hence the $\Lambda$-list property is satisfied for every $e \in E_U$. Since \(\{1,2, \dots, |\mathcal{F}| \} = \Lambda(a) = \Lambda(z_{t+1}) = \dots = \Lambda(z_{|\mathcal{F}|})\), the edges in $E_{Za}$ also satisfy the $\Lambda$-list property.

We now label the edges in $E_{Ua}$ by defining the sets $S^1_L, \dots, S^t_L$. Recall that $\mathcal{F}' = \{F_1, \dots, F_t \}$ is a set cover of size $t$. We set
\begin{align*}
S^1_L := \{ \{u,a\} \mid u \in F_1 \} \text{ and } S^i_L~:=~\{ \{u,a\} \mid u \in F_i \setminus (F_1 \cup \dots \cup F_{i-1}) \}\text{ for each } i \in \{2, 3,\dots,t\}.
\end{align*}
Obviously, each edge of $E_{Ua}$ is an element of at most one of the sets $S^1_L, \dots, S^t_L$. Since $\mathcal{F'}$ is a set cover, we know that $F_1 \cup \dots \cup F_t = U$. Hence, $S^1_L, \dots ,S^t_L$ is a partition of~$E_{Ua}$. Since $U \cup \{ a\}$ forms a clique, no pair of edges in $E_{Ua}$ violates the definition of STC-labelings.
Moreover, since the edges in \(E_{Ua}\) receive different colors from the edges in \(E_{Za}\), no pair of edges from these two sets violates the definition of STC-labelings.
From the definition of $\Lambda$ we know that $\Lambda(a) = \{ 1, \dots, |\mathcal{F}| \}$ and for every $u \in U$ it holds that $i \in \Lambda(u)$ if $u \in F_i$. Hence, every edge in $E_{Ua}$ satisfies the $\Lambda$-list property. It follows that $L$ is a $c$-colored STC-labeling with $W_L = \emptyset$ such that every edge satisfies the $\Lambda$-list property under $L$.

$(\Leftarrow)$ Conversely, let $L~=~(S^1_L, \dots, S^{|\mathcal{F}|+1}_L, \emptyset)$ be a $c$-colored STC-labeling for $G$ such that every edge of $G$ satisfies the $\Lambda$-list property. We will construct a set cover $\mathcal{F}' \subseteq \mathcal{F}$ with $|\mathcal{F}'| \leq t$. We focus on the vertex $a$ and its incident edges. Those are exactly the edges of $E_{Ua} \cup E_{Za}$. Since there are no weak edges, we know that all those edges are elements of strong color classes $S^i_L$. Since $L$ is an STC-labeling and every pair of edges $e, e' \in E_{Za}$ forms a $P_3$, it follows by $|E_{Za}|= |\mathcal{F}|-t$ that the edges in \(E_{Za}\) are elements of $|\mathcal{F}|-t$ distinct color classes. Because there is no edge between the vertices of $U$ and $Z$, it also holds that there is no $e \in E_{Ua}$ that is an element of the same strong color class as some $e' \in E_{Za}$. Otherwise, $e$ and $e'$ form a $P_3$ with the same strong color which contradicts the fact that $L$ is an STC-labeling. It follows that the edges in $E_{Ua}$ are elements of at most $t$ distinct strong color classes, since $|\Lambda(a)|=|\mathcal{F}|$ and every edge of $G$ satisfies the $\Lambda$-list property under $L$. Without loss of generality (by renaming) we can assume that those strong color classes are $S^1_L, \dots, S^t_L$. Recall that $\mathcal{F}=\{F_1,F_2,\dots,F_{|F|}\}$. We define
\begin{align*}
\mathcal{F}' := \lbrace F_1, F_2, \dots, F_t \rbrace \text{.}
\end{align*}
Obviously, $|\mathcal{F}'|=t$. It remains to show that $\mathcal{F}'$ is a set cover.  The fact that all edges of $G$ satisfy the $\Lambda$-list property under $L$, implies that for every $u \in U$ there is a  $j \in \{1, \dots, t\}$ such that $j \in \Lambda(u)$. Since $\Lambda(u) = \{i \mid u \in F_i\} \cup \{|\mathcal{F}|+1\}$ for all $u \in U$ by construction, it follows that every $u \in U$ is an element of one of the sets $F_1, F_2, \dots, F_t$. Hence,~$U = F_1 \cup F_2 \cup \dots \cup F_t$ and, thus, $\mathcal{F}'$ is a set cover of size $t$.
% \end{corollary}
\end{proof}

A closer look at the instance $(G,c,k,\Lambda)$ for \textsc{VL-Multi-STC} constructed from the  \textsc{Set Cover} instance~$(U, \mathcal{F}, t)$ in the proof of Theorem~\ref{VertList-W1-by-ko} reveals that $c=|\mathcal{F}|+1$ and $\ko \leq |\mathcal{F}| - t$. It follows that $c+\ko \leq 2|\mathcal{F}|+1$, so the construction is a polynomial parameter transformation from \textsc{Set Cover} parameterized by~$|\mathcal{F}|$ to \textsc{VL-Multi-STC} parameterized by $(c,\ko)$. Now, since \textsc{Set Cover} parameterized by $|\mathcal{F}|$ does not admit a polynomial kernel unless $\badStuffHappens$ \cite{DomLS14} we obtain the following.
\begin{corollary} \label{Corollary: VL-Mult-STC no poly kernel}
\textsc{VL-Multi-STC} parameterized by $(c, \ko)$ does not admit a polynomial kernel unless~$\badStuffHappens$.
\end{corollary}

\subsection{On Problem Kernelization}
\label{sec:kernel}

Since \textsc{EL-Multi-STC} is a generalization of \textsc{VL-Multi-STC}, Corollary~\ref{Corollary: VL-Mult-STC no poly kernel} implies that there is no polynomial kernel for \textsc{EL-Multi-STC} parameterized by $(c,\ko)$ unless $\badStuffHappens$.
We now show that there is a $2^{c+1} \cdot \ko$-vertex kernel for \textsc{EL-Multi-STC}. To this end, we define a new parameter~$\tau$ as follows. Let $I:=(G,c,k,\Psi)$ be an instance of \textsc{EL-Multi-STC}. Then~$\tau:= |\Psi(E) \setminus \{\emptyset\}|$ is defined as the number of different non-empty edge lists occurring in the instance $I$. It clearly holds that~$\tau \leq 2^c-1$.

 For this kernelization we use \emph{critical cliques} and \emph{critical clique graphs}~\cite{Protti2009}.  These concepts were also used to obtain linear-vertex kernels for graph clustering problems parameterized by the number of edge modifications~\cite{Guo09,CM12} and for STC~\cite{GK20} parameterized by the number of weak edges~$k$. Since for STC we can assume that~$k=\ko$, the kernelization described here lifts this linear-vertex kernel for~\textsc{STC} to the more general~\textsc{EL-Multi-STC}.

\begin{definition} \label{Definition: Critical Cliques}
A \emph{critical clique} of a graph $G$ is a clique $K$ where the vertices of $K$ all have the same neighbors in $V \setminus K$, that is, \(\forall u, v \in K \colon N(u)\setminus K = N(v) \setminus K\), and $K$ is maximal under this property.
%\end{definition}
%
%\begin{definition}
Given a graph $G=(V,E)$, let $\mathcal{K}$ be the collection of its critical cliques. The \emph{critical clique graph} $\mathcal{C}$ of~$G$ is the graph $(\mathcal{K},E_\mathcal{C})$ with
%\begin{align*}
$\lbrace K_i, K_j \rbrace \in E_{\mathcal{C}} \Leftrightarrow \forall u \in K_i, v \in K_j: \lbrace u , v \rbrace \in E$ \text{.}
%\end{align*}
\end{definition}

For a critical clique $K$ we let $\mathcal{N}(K) := \bigcup_{K' \in N_\mathcal{C}(K)} K'$ denote the union of its neighbor cliques in the critical clique graph and $\mathcal{N}^2 (K) := \bigcup_{K' \in N_\mathcal{C}^2(K)} K'$ denote the union of the critical cliques at distance exactly two from~$K$. The critical clique graph can be constructed in $\mathcal{O}(n+m)$ time~\cite{Hsu}.

Critical cliques are an important tool for \textsc{EL-Multi-STC}, because every edge contained in some critical clique is not part of any induced $P_3$ in $G$. Hence, each such edge~$e$ is strong under any STC-labeling unless~$\Psi(e)=\emptyset$. In the following, we will distinguish between two types of critical cliques. We say that a critical clique $K$~is \emph{closed} if $\mathcal{N}(K)$ forms a clique in $G$ and that~$K$ is \emph{open} otherwise. The number of vertices in open critical cliques is at most $2\ko$~\cite{GK20}. The \iflong main reduction rule of this kernelization describes how to deal with large closed critical cliques. Before we give the concrete rules we provide a useful property of closed critical cliques. Informally, Lemma~\ref{Lemma: Type two neighbour monochrome} states that edges between a closed critical cliques and a neighbor vertex all behave the same if they have the same strong color~lists.

\begin{restatable}{lemma}{typetwoneighbor}
  % \begin{proposition}
    \label{Lemma: Type two neighbour monochrome} Let $G=(V,E)$ be a
    graph, let $K$ be a closed critical clique in~$G$, and let
    $\Psi: E \rightarrow 2^{\{1,\dots,c\}}$ a mapping for some $c \in \mathds{N}$. Moreover, let~$v \in \mathcal{N}(K)$ and~$E' \subseteq E(\{v\}, K)$ such that all $e' \in E'$ have the
    same strong color list under~$\Psi$. Then, there is an optimal STC-labeling
    $L=(S^1_L, S^2_L, \dots, S^c_L, W_L)$ for $G$ and~$\Psi$ such that $E' \subseteq A$ for
    some $A \in \{ S^1_L, \dots, S^c_L, W_L \}$.
  % \end{proposition}
\end{restatable}

\begin{proof}
Pick an optimal STC-labeling $L=(S^1_L, S^2_L, \dots, S^c_L, W_L)$. If~$E' \subseteq W_L$ nothing more needs to be shown. Otherwise, if $E' \not \subseteq W_L$, there exists an edge~$e\in E'$ with~$e \in S^i_L$ for some~$i \in \{1, \dots, c\}$.

We then define a labeling~$\hat{L} = (S^1_{\hat{L}}, S^2_{\hat{L}}, \dots, S^c_{\hat{L}}, W_{\hat{L}})$ by~$S^i_{\hat{L}} := S^i_L \cup E'$, $W_{\hat{L}} := W_L \setminus E'$ and~$S^j_{\hat{L}} := S^j_L \setminus E'$ for~$j \neq i$ and show that~$\hat{L}$ is a~$\Psi$-satisfying optimal STC-labeling.

Let~$\overline{e} \in E$ be an edge that forms an induced~$P_3$ with some edge in~$E'$. Since $K \cup \mathcal{N}(K)$ forms a clique by the definition of closed critical cliques, it follows that $\overline{e} \in E(\{v\}, \mathcal{N}^2(K))$. Note that $\overline{e}$ also forms an induced $P_3$ with $e$. Then, $\overline{e} \not \in S^i_{\hat{L}}$ since $L$ is an STC-labeling and $e \in S^i_L$. Hence, $\hat{L}$ does not violate STC. Moreover, it holds that~$|W_{\hat{L}}| \leq |W_L|$ and thus,~$\hat{L}$ is optimal.

From the definition of $E'$ and the fact that $e \in S^i_L$, we know
that $i \in \Psi(e')$ for all~$e' \in E'$. Hence, $\hat{L}$ is $\Psi$-satisfying.
\end{proof}

\begin{algorithm} [t]
\caption{EL-Multi-STC kernelization} \label{Algorithm: EL-Multi-STC kernel reduction}
\begin{algorithmic}[1]
\State \textbf{Input:} A graph $G=(V,E)$ graph and a closed critical clique~$K \subseteq V$ in $G$
\For {\textbf{each} $v \in \mathcal{N}(K)$}
\For {\textbf{each} $\psi \in \{ \Psi(e) \neq \emptyset \mid e \in E(\{v\},K) \}$}
\State $i := 0$
\For {\textbf{each} $w \in N(v) \cap K$}
\If {$\Psi(\{v,w\}) = \psi$}
\State Mark $w$ as \emph{important}  \label{Line: Mark as important}
\State $i:=i+1$
\EndIf
\If {$i = |E(\{v\},\mathcal{N}^2(K))|$}  \label{Line: important bound 1}
\State \textbf{break} \label{Line: important bound 2}
\EndIf
\EndFor
\EndFor
\EndFor
\State Delete all vertices $u \in K$ which are not marked as important from $G$
\State Decrease the value of $k$ by the number of edges $e$ such that~$e$ is incident with a deleted vertex $u$ and~$\Psi(e)= \emptyset$.
\end{algorithmic}
\end{algorithm}

We now may use Lemma~\ref{Lemma: Type two neighbour monochrome} in the following way.
We take a critical clique~\(K\) and for each vertex~\(v \in\mathcal{N}(K)\) and each color list we take an edge between \(v\) and \(K\) with that color list and we mark the other end of that edge.
Then we remove all vertices in \(K\) that are not marked, see Algorithm~\ref{Algorithm: EL-Multi-STC kernel reduction}.
The remaining vertices (and edges) represent the deleted vertices and using Lemma~\ref{Lemma: Type two neighbour monochrome} we can show that the deleted edges behave the same as the remaining ones.
This gives rise to the following reduction rule.

\begin{redrule} \label{Rule: T2 Cliques big kernel}
If $G$ has a closed critical clique $K$ with  $|K| > \tau \cdot |E(\mathcal{N}(K),\mathcal{N}^2(K))|$, then apply Algorithm~\ref{Algorithm: EL-Multi-STC kernel reduction} on $G$ and $K$.
\end{redrule}
\begin{restatable}{proposition}{algo}
  % \begin{proposition}
    \label{Proposition: Kernel Algorithm}
    Rule~\ref{Rule: T2 Cliques big kernel} is safe and can be applied in polynomial time.
  % \end{proposition}
\end{restatable}

\begin{proof} %[of Proposition \ref{Proposition: Kernel Algorithm}]
Let $(G=(V,E), c, k , \Psi)$ be an
    instance for \textsc{EL-Multi-STC} and let $K$ be a closed critical clique. We show that Algorithm~\ref{Algorithm: EL-Multi-STC kernel reduction} applied on $G$ and $K$ runs in
    $\mathcal{O}(n^3)$ time and produces an equivalent instance $(G'=(V',E'), c, k', \Psi')$
    for \textsc{EL-Multi-STC}.

Since $|\mathcal{N}(K)|, |N(v)\cap K| \leq n$ and $|\{ \Psi(e) \neq \emptyset \mid e \in E(\{v\},K) \}| \leq |K| \leq n$, the given algorithm clearly runs in $\mathcal{O}(n^3)$ time. It remains to show that the produced instance $I':=(G'=(V',E'), c, k', \Psi')$ is equivalent to $I:=(G=(V,E), c, k , \Psi)$. Let $D_V \subseteq V$ be the set of vertices that were deleted by Algorithm \ref{Algorithm: EL-Multi-STC kernel reduction}, let $D_E$ be the set of edges that are incident with some $v \in D_V$ and let $D^{\emptyset}_E \subseteq D_E$ be the set of edges $e \in D_E$ with $\Psi(e)=\emptyset$. We have
\begin{align*}
G'=(V \setminus D_V, E \setminus D_E) \text{; }  k'= k-|D_E^{\emptyset}| \text{; and } \Psi'= \Psi|_{E \setminus D_E} \text{.}
\end{align*}
We also define $K' := K\setminus D_V$ as the modified critical clique in $G'$.

$(\Rightarrow)$ Let $L=(S^1_L, S^2_L, \dots, S^c_L, W_L)$ be a $\Psi$-satisfying STC-labeling for $G$ such that $|W_L| \leq k$. We define a labeling $\hat{L}= (S^1_{\hat{L}},\dots,S^c_{\hat{L}},W_{\hat{L}})$ by $W_{\hat{L}}:=W_L \setminus D_E$ and $S_{\hat{L}}^i := S^i_L \setminus D_E$ for each $i \in \{1, \dots, c\}$. The fact that $L$ is $\Psi$-satisfying implies that $\hat{L}$ is $\Psi'$-satisfying. It also holds that
\begin{align*}
|W_{\hat{L}}| = |W_L \setminus D_E| = |W_L| - |W_L \cap D_E| \leq k - |D_E^{\emptyset}| = k' \text{,}
\end{align*}
since $D_E^{\emptyset} \subseteq W_L \cap D_E$. It remains to prove that $\hat{L}$ does not violate~STC. Assume there is an induced $P_3$ on vertices $u,v$, and~$w$ in~$V'$ with edges $\{ u, v \} ,\{v,w\} \in S^i_{\hat{L}}$ for some $1 \leq i \leq c$. It follows that $\{u,w\} \in D_E$, since $L$ is an STC-labeling. Then, by the definition of $D_E$, at least one of the vertices $u$ or $w$ was deleted by the algorithm. This contradicts the fact that $u,w \in V'=V \setminus D_V$. It follows that $\hat{L}$ is a $\Psi'$-satisfying STC-labeling for $G'$ with at most $k'$ weak edges.

$(\Leftarrow)$ Conversely, let $\hat{L}= (S^1_{\hat{L}},\dots,S^c_{\hat{L}},W_{\hat{L}})$ be a $\Psi'$-satisfying STC-labeling for $G'$ such that $|W_{\hat{L}}| \leq k - |D_E^{\emptyset}|$. We define a $\Psi$-satisfying STC-labeling $L$ for $G$, with $|W_L| \leq k$. We first show that if a vertex in~$\mathcal{N}(K')$ is connected to~$K'$ by many edges with the same strong color list, then these edges may receive the same strong color. To show the claim, consider a fixed vertex $v \in \mathcal{N}(K')$ and a set $K_v \subseteq K'$ such that all edges in $E(\{v\}, K_v)$ have the same strong color list $\psi \neq \emptyset$ under $\Psi'$. %From the lines $8$ and $9$ of the algorithm, we know that $|K_v| \leq |

\begin{cla} \label{Claim: Strong neighbors}
If $|K_v| \geq |E(\{v\},\mathcal{N}^2(K'))|$, then we can assume that $E(\{v\}, K_v) \subseteq S^i_{\hat{L}}$ for some $1 \leq i \leq c$.
\end{cla}

\begin{claimproof}
Since $K'$ is a closed critical clique, we can assume by Lemma \ref{Lemma: Type two neighbour monochrome}, that either all edges in~$E(\{v\}, K_v)$ are weak or have the same strong color under $\hat{L}$. In the latter case, the claim is directly fulfilled. Thus, assume the first case holds, that is,~$E(\{v\}, K_v) \subseteq~W_{\hat{L}}$.
Note that, whenever an edge $e \in E(\{v\}, K_v)$ forms an induced $P_3$ with another edge $e'$, it follows that $e' \in E(\{v\},\mathcal{N}^2(K'))$. Let $i \in \psi$. We define a new labeling $P=(S^1_P, \dots, S^c_P, W_P)$ for $G'$ by
\begin{align*}
S^i_P &:= S^i_{\hat{L}} \cup E(\{v\}, K_v) \setminus E(\{v\},\mathcal{N}^2(K'))\text{,}\\
W_P &:= W_{\hat{L}} \setminus E(\{v\}, K_v) \cup (S^i_{\hat{L}} \cap E(\{v\},\mathcal{N}^2(K')))  \text{, and}\\
S^j_P &:= S^j_{\hat{L}} \text{ for all }j \neq i \text{.}
\end{align*}
From $|K_v| \geq |E(\{v\},\mathcal{N}^2(K'))|$ we conclude
\begin{align*}
|W_P| &= |W_{\hat{L}}| - |E(\{v\}, K_v)| +|S^i_{\hat{L}} \cap E(\{v\},\mathcal{N}^2(K'))|\\
& \leq |W_{\hat{L}}| - |K_v| + |E(\{v\},\mathcal{N}^2(K'))|\\
& \leq |W_{\hat{L}}| \text{.}
\end{align*}
Moreover, $P$ clearly is $\Psi'$-satisfying. It remains to show that $P$ is an STC-labeling, which means that there is no induced $P_3$ containing an edge $e \in E(\{v\},K_v) \subseteq S^i_P$ and another edge $e' \in S^i_P$. As mentioned above, the edges in $E(\{v\},K_v)$ only form an induced $P_3$ with edges from $E(\{v\},\mathcal{N}^2(K'))$. By the construction of $P$, no edge from $E(\{v\},\mathcal{N}^2(K'))$ belongs to $S^i_P$. Hence, $P$ is an STC-labeling.~$\hfill \Diamond$
\end{claimproof}

Let thus \(\hat{L}\) be such that the assumption of Claim~\ref{Claim: Strong neighbors} holds simultaneously for all eligible \(v \in \mathcal{N}(K')\) and \(K_v\).
We define the labeling $L$ for $G$ by extending $\hat{L}$. We set $W_L := W_{\hat{L}} \cup D_E^{\emptyset}$. Since $|W_{\hat{L}}| \leq k-|D_E^{\emptyset}|$, it holds that $|W_L| \leq k$. It remains to label all edges in $D_E \setminus D_E^{\emptyset}$. Let $u$ be some fixed vertex in $D_V$ and $v \in N(u)$ such that $\{u,v\} \not \in D_E^{\emptyset}$.

\textbf{Case 1:} If $v \in K$, then the edge $\{u,v\}$ is an edge between two vertices of a critical clique. Since~$\{u,v\} \not \in D_E^{\emptyset}$, there is some $i \in \Psi(\{u,v\})$. We add $\{u,v\}$ to $S^i_L$.
Clearly, $\{u,v\}$ satisfies the $\Psi$-list property.
Since $\{u,v\}$ is not part of any induced~$P_3$ this does not violate STC.

\textbf{Case 2:} If $v \in \mathcal{N}(K)$, then there is a set $Y \subseteq K'$ containing at least $|E(\{v\},\mathcal{N}^2(K))|$ vertices distinct from $u$ such that $\Psi(\lbrace v, y\rbrace) =  \Psi(\{v,u\})$ for every $y \in Y$. Otherwise, $u$ would have been marked as \emph{important} by Algorithm~\ref{Algorithm: EL-Multi-STC kernel reduction}, which contradicts the fact that $u \in D_V$. From Claim \ref{Claim: Strong neighbors}, we know that all edges in $E_{G'}(\{v\},Y)$ are elements of the same strong color class $S^i_{\hat{L}}$ for some $i \in \{1, \dots, c \}$. We set $S^i_L := S^i_{\hat{L}} \cup \{\{u,v\}\}$. Clearly, $\{u,v\}$ satisfies the $\Psi$-list property under $L$. Moreover, adding~$\{u,v\}$ to~$S^i_L$ does not violate STC: The only edges that form an induced~$P_3$ with~$\{u,v\}$ are the edges in $E(\lbrace v \rbrace, \mathcal{N}^2(K))$. Now, since~$E_{G'}(\{v\},Y) \subseteq S^i_{\hat{L}}$ and since~$\hat{L}$ is an STC-labeling for~$G'$, none of these edges is contained in $S^i_{\hat{L}}$.

Consequently, $L$ is a $\Psi$-satisfying STC-labeling for $G$ with~$|W_L|~\leq~k$.
\end{proof}

%Rule \ref{Rule: T2 Cliques big kernel} leads to the following kernel result.
We now consider instances that are reduced regarding Rule~\ref{Rule: T2 Cliques big kernel}. The following upper bound of the size of closed critical cliques is important for the kernel result.

 \begin{restatable}{lemma}{ttwoupperbound}
  % \begin{proposition}
  \label{Lemma: Upper Bound for T2 cliques}
  Let $(G,c,k,\Psi)$ be a reduced instance for \textsc{EL-Multi-STC}. For every closed critical
  clique~$K$ in $G$ it holds that
  %\begin{align*}
    $|K| \leq \tau \cdot |E(\mathcal{N}(K),\mathcal{N}^2(K))|$.
  %\end{align*}
  % \end{proposition}
\end{restatable}

\begin{proof}
We prove the lemma by having a closer look at the vertices that were not deleted by Algorithm~\ref{Algorithm: EL-Multi-STC kernel reduction}. Note that the algorithm is applied on every closed critical clique $K$ with $|K| >  \tau \cdot |E(\mathcal{N}(K),\mathcal{N}^2(K))|$. Every vertex that was not marked as \emph{important} in Line~\ref{Line: Mark as important} of the algorithm is deleted from $G$. Note that there are at most $\tau$ possible images $\psi$ of $\Psi: E \rightarrow 2^{\{1,\dots,c\}}$. By Lines~\ref{Line: Mark as important},~\ref{Line: important bound 1}, and \ref{Line: important bound 2} it holds that for every $v \in \mathcal{N}(v)$ the algorithm marks at most $\tau \cdot |E(\{v\},\mathcal{N}^2(K))|$ vertices of $K$. Consequently, there are at most
\begin{align*}
\tau \cdot \sum_{v \in \mathcal{N}(K)} |E(\{v\},\mathcal{N}^2(K))| = \tau \cdot |E(\mathcal{N}(K),\mathcal{N}^2(K))|
\end{align*}
marked vertices in $K$, since $\{ E(\{v\},\mathcal{N}^2(K)) \mid v \in \mathcal{N}(v) \}$ forms a partition of $E(\mathcal{N}(K), \mathcal{N}^2(K))$. Hence, $|K| \leq \tau\cdot |E(\mathcal{N}(K),\mathcal{N}^2(K))|$ for every closed critical clique $K$ in $G$.
\end{proof}

\begin{restatable}{theorem}{kernelsize}
%\begin{theorem}
\label{Theorem: (k1,c) Kernel for EL-Multi-STC}
\textsc{EL-Multi-STC} admits a problem kernel with at most $(\tau+1)~\cdot~2\ko$ vertices.
\end{restatable}

\begin{proof} 
Let $(G=(V,E),c,k,\Psi)$ be a reduced instance for \textsc{EL-Multi-STC}. We show that~$|V| \leq (\tau+1)~\cdot~2\ko$.

The overall number of vertices in open critical cliques is at most $2\ko$~\cite{GK20}. Let~$K$ be some closed critical clique. We now transform the graph $G$ into a modified graph ${G}'$ in the following way. 
We replace every closed critical clique $K$ with a critical clique $K'$ such that $|{K}'|=\frac{|K|}{\tau}$. From Lemma~\ref{Lemma: Upper Bound for T2 cliques} we know that $|K| \leq \tau \cdot |E(\mathcal{N}(K),\mathcal{N}^2(K))|$  for every closed critical clique $K$ in $G$. Consequently,~$|{K'}| \leq |E(\mathcal{N}(K),\mathcal{N}^2(K))|$ for every closed critical clique ${K'}$ in ${G'}$. As shown previously, this implies that the overall number of vertices in closed critical cliques in ${G'}$ is at most~$2 \ko$~\cite[Proof of Theorem~1]{GK20}. Hence, the overall number of vertices in closed critical cliques in $G$ is at most $\tau \cdot 2 \ko$, which gives us
%\begin{align*}
$|V| \leq 2k_1 + \tau \cdot 2 \ko = (\tau+1) \cdot 2\ko$.
%\end{align*}
\end{proof}

\iflong Recall that for \else For the last two statements of Theorem~\ref{Theorem: (c+1)^ko Algorithm for EL-M-STC}, recall that for \fi any \textsc{EL-Multi-STC} instance $(G,c,k,\Psi)$ we have~$\tau \leq 2^c -1$. Also, \textsc{Multi-STC} is the special case of \textsc{EL-Multi-STC} where every edge has the list $\{1,2,\dots,c\}$, and thus~$\tau=1$. These two facts imply the following.

\begin{corollary} \label{Corollary: k1 kernel}
\textsc{EL-Multi-STC} admits a problem kernel with at most $2^{c+1} \ko$ vertices. \textsc{Multi-STC} admits a problem kernel with at most $4 \ko$ vertices. 
\end{corollary}

\section{Conclusion}We have provided a first study of the computational complexity of
\textsc{Multi-STC} and two of its generalizations.  There are many interesting research
questions that can be pursued in future work. Most importantly, it is open whether
\textsc{Multi-STC} admits an algorithm with running time~$2^{\Oh(n)}$, even when~$k=0$. Let
us remark that an algorithm with such a running time is also open for \textsc{Edge Coloring}:
so far, it is only known that \textsc{List-Edge Coloring}, where a color list is given for each
edge, admits no~$2^{o(n^2)}$-time algorithm under the ETH. A first step to prove a lower
bound for \textsc{Edge Coloring} could be to prove that a vertex-list version of \textsc{Edge
  Coloring} admits no~$2^{o(n^2)}$-time algorithm as well. 
  
For \textsc{Multi-STC} with~$c=2$, there exists a polynomial kernel for the parameter~$k$: There is a parameter preserving reduction to \textsc{Odd Cycle Transversal}~\cite{ST14} which has a polynomial kernel~\cite{KW14}. Then, since \textsc{Multi-STC} with~$c=2$ is NP-hard and \textsc{Odd Cycle Transversal} is in NP, there exists a polynomial kernel for Multi-STC. However, an interesting open question is if one could think of any direct problem kernelization for the parameter~$k$.

A further direction of research could be to identify other applications of the structural
parameter~$k_1$ which is the vertex cover number of the Gallai graph. Even more generally, it
seems interesting to study which graph problems become tractable when the Gallai graph of the
input graph has a certain structure. To this end, one could further study relations between the
structure of a graph and the structure of its Gallai graph. For example, is it possible to
describe certain natural graph classes compactly via forbidden induced subgraphs of their
Gallai graphs?

\end{document}